\documentclass{article}
\usepackage{arxiv}
\usepackage{listings}
\usepackage{xcolor}
\usepackage{amssymb}
\usepackage{amsmath}
\usepackage{amsthm}
\usepackage{mathtools}
\usepackage{hyperref}
\usepackage{cleveref}

\theoremstyle{plain}
\newtheorem{theorem}{Theorem}[section]
\newtheorem{proposition}[theorem]{Proposition}
\newtheorem{lemma}[theorem]{Lemma}

\theoremstyle{definition}
\newtheorem{definition}[theorem]{Definition}
\newtheorem{setup}[theorem]{Setup}
\newtheorem{remark}[theorem]{Remark}
\numberwithin{equation}{section}
\crefname{setup}{Setup}{Setups}

\input{preamble.sty}

\begin{document}

\title{Partially deterministic sampling for compressed sensing with denoising guarantees}

\author{Yaniv Plan\thanks{Department of Mathematics, University of British Columbia, Vancouver, BC, Canada.}
\and Matthew S. Scott\footnotemark[1]
\and Ozgur Yilmaz\footnotemark[1]\thanks{CNRS -- PIMS International Research Laboratory}}

\begin{abstract}
We study compressed sensing when the sampling vectors are chosen from the rows of a unitary matrix.  In the literature, these sampling vectors are typically chosen randomly; the use of randomness has enabled
major empirical and theoretical advances in the field. However, in practice there are often certain crucial sampling vectors, in which case practitioners will depart
from the theory and sample such rows deterministically. 
In this work, we derive an optimized sampling scheme
for Bernoulli selectors which naturally combines random and deterministic
selection of rows, thus rigorously deciding which rows should be sampled deterministically.
This sampling scheme provides measurable improvements in
image compressed sensing for both generative and sparse priors when compared
to with-replacement and without-replacement sampling schemes, as we show
with theoretical results and numerical experiments. Additionally, our
theoretical guarantees feature improved sample complexity bounds compared
to previous works, and novel denoising guarantees in this setting.
\end{abstract}
\keywords{Compressed sensing; Bernoulli sampling; Optimal sampling; Denoising; Generative models.}

\maketitle
\begingroup
\renewcommand{\thefootnote}{}
\NoHyper
\footnotetext{{\bf Author roles:} Authors listed in alphabetic order. MS is the lead author of this manuscript.}
\endNoHyper
\addtocounter{footnote}{-1}
\endgroup

\section{Introduction}
\label{loc:body.introduction}
Compressed sensing (CS) enables the recovery of high-dimensional signals with low-dimensional structure from far fewer measurements than their ambient dimension, a paradigm that led to ``compressive signal acquisition" and has transformed fields like medical and seismic imaging, computational photography, and radar and remote sensing. 
Consider a signal $\boldsymbol{x}_0 \in \mathbb{K}^n$ (where $\mathbb{K}$ is either $\mathbb{R}$ or $\mathbb{C}$) that lies in or near a prior set $\mathcal{Q} \subseteq \mathbb{K}^n$ with effective dimensionality much smaller than $n$. 
We aim to reconstruct $\boldsymbol{x}_0$ from noisy measurements of the form $\boldsymbol{y} = A \boldsymbol{x}_0 + \boldsymbol{\epsilon}$, where $A \in \mathbb{K}^{m \times n}$ is a CS matrix with $m \ll n$, and $\boldsymbol{\epsilon} \sim \mathcal{N}(0, \sigma^2 I)$ represents Gaussian noise. 
We consider subsampled unitary CS matrices, which are of the form $A = SF$ for $F \in \mathbb{K}^{n \times n}$ a unitary matrix (e.g., Fourier), and $S \in \mathbb{R}^{m \times n}$ a random \emph{sampling matrix}, a matrix each row of which has exactly one entry that is non-zero and identical. We call the rows of $F$ \emph{measurement vectors}. The sampling matrix $S$ specifies a selection (and a uniform scaling) of these measurement vectors, and each selected measurement vector performs a noisy measurement on the true signal $\boldsymbol{x}_0$ via a noisy inner product.

This model matches applications like magnetic resonance imaging (MRI), where measurements are constrained to the Fourier domain~\cite{lustigSparseMRIApplication2007, chauffertVariableDensitySampling2014}. 
Early CS research recognized that not all measurements are equally informative; for example, low-frequency Fourier components are critical for natural images. 
This insight led to the development of \emph{optimized sampling schemes}, also known in the literature as \emph{optimal} or \emph{near-optimal} sampling schemes, where measurements are prioritized based on their local coherences. (Local coherences quantify the alignment of each measurement vector with the prior set $\mathcal{Q}$~\cite{candesSparsityIncoherenceCompressive2007, puyVariableDensityCompressive2011, krahmerStableRobustSampling2014}.)
Optimized sampling schemes fall into the broader category of \emph{variable-density sampling schemes}, which are randomized sampling schemes parameterized by probability vectors.

The simplest class of sampling distributions is with-replacement sampling,
where measurement vectors are repeatedly drawn independently at random
according to some fixed probability vector.
Under this scheme, even the measurement vectors with the highest sampling probabilities can still be entirely missed, a phenomenon that can be significant in certain edge-case settings, as we show in a toy example in \Cref{loc:body.toy_example}. 

In an effort to address this limitation, Puy, Vandergheynst and Wiaux~\cite{puyVariableDensityCompressive2011} considered variable-density sampling with Bernoulli selectors, where each measurement vector is included according to its own independent Bernoulli random variable. For this type of sampling distribution, it is possible
to set the probability weight of certain measurement vectors to $1$, thereby including them deterministically. Bernoulli sampling schemes therefore naturally incorporate deterministic measurements into random sampling schemes.
Block sampling emerged as another solution, where measurement vectors are arranged into blocks, and different blocks have their own sampling distributions~\cite{bigotAnalysisBlockSampling2016, polakPerformanceBoundsGrouped2015}. This paradigm allows for blocks of deterministically sampled measurements to be incorporated in a structured fashion. 

Beyond the risk of entirely missing crucial measurements, with-replacement sampling has a second drawback: it allows for the repeated sampling of some measurement vectors, which yield no additional information about the true signal beyond their limited denoising effect. Both Bernoulli and without-replacement sampling are means of avoiding redundant selections~\cite{hoppeSamplingStrategiesCompressive2023}.
However, without-replacement schemes still cannot guarantee the inclusion of high-coherence measurement vectors, which may be missed with non-negligible probability in certain regimes.
Optimized Bernoulli sampling stands out as a powerful solution, avoiding both redundant rows and missed critical measurements, while still allowing the sampling distribution to adapt to the local coherences.
This hybrid approach is particularly advantageous in settings like MRI, where some low-frequency measurements are known to be essential, and their inclusion should therefore not be left up to chance. 
Moreover, Bernoulli sampling is widely adopted in CS literature due to its simplicity and compatibility with theoretical analysis (e.g., \cite{rudelsonSparseReconstructionFourier2008}).

In this paper, we advance the theory and practice of optimized Bernoulli sampling in a number of ways. 
\begin{itemize}
\item \textbf{Denoising guarantees.} 
    We establish the first denoising bounds for Bernoulli sampling under Gaussian noise and variable probability weights, 
    extending the with-replacement theory of~\cite{planDenoisingGuaranteesOptimized2025} 
    to this setting.
\item \textbf{Closed-form probability weights.} 
We give a fully explicit expression for the optimized Bernoulli probability weights (\Cref{loc:optimized_bernoulli_weights.statement}), avoiding the need to solve an auxiliary equation for a normalizing constant as in~\cite{adcockUnifiedFrameworkLearning2024} or an optimization problem as in~\cite{puyVariableDensityCompressive2011}. The closed-form expression simplifies implementation and makes the dependence on the local coherences and the number of measurements transparent.
\item \textbf{Improved sample-complexity bounds.} 
    Our bounds improve on those of optimized with-replacement sampling~\cite{planDenoisingGuaranteesOptimized2025} by replacing
    $\|\boldsymbol{\alpha}\|_2$ with $L(\boldsymbol{\alpha},m)$ in the sample complexity and noise bound (see \Cref{loc:local_coherence.statement} and \Cref{loc:optimized_bernoulli_weights.statement} for the definitions of $\boldsymbol{\alpha}$ and $L$ respectively).  We show that $L(\boldsymbol{\alpha},m) \leq \|\boldsymbol{\alpha}\|_2$; the enhancement comes from removing the dependence of the sample complexity on the local coherences of saturated measurements. When the number of measurements $m$ is large, we demonstrate in~\Cref{fig:numerical_complexity_bounds.pdf} that the improvement in the bound can be empirically significant in realistic settings.   
\item \textbf{General recovery guarantees for arbitrary Bernoulli weights.} We provide a recovery and denoising theory for any Bernoulli selector sampling scheme, with the optimized scheme obtained as the minimizer of the associated complexity bound. 
\end{itemize}

\textbf{Motivation for Bernoulli selectors.}
Bernoulli sampling schemes occupy a useful middle ground: they allow crucial rows to be included deterministically when their probability weights equal 1, while still enabling flexible randomization for the remaining measurements. This flexibility makes Bernoulli selectors especially attractive in applications such as MRI, where certain low-frequency measurements are indispensable and should not be left to chance.

\textbf{Structure and flow.}
Section~\ref{loc:body.main_result} presents our main recovery guarantee,
\Cref{loc:optimized_bernoulli_cs_on_union_of_subspaces.statement}, together with the
optimized Bernoulli weights $\boldsymbol{w}^\circ$ and the associated quantity $L(\boldsymbol{\alpha},m)$ that
appears in the sample-complexity bound. The theoretical framework leading to this result
is developed in the subsequent sections. This ordering is intentional: it allows the 
optimized sampling scheme and its consequences to be seen upfront, while the 
subsequent sections gradually build the theoretical framework that justifies it.
With this in mind, the rest of the paper is organized as follows: 
Section~\ref{loc:body.general_signal_recovery} develops a general RIP-based recovery
theory for an arbitrary sampling matrix: it introduces the truncation operator, establishes
a signal recovery bound under an RIP assumption (\Cref{loc:signal_recovery_with_subsampled_unitary_matrix_with_gaussian_noise_on_union_of_subspaces.bernoulli}), and derives general upper
bounds on the noise term $\| \tilde D\,T(SD\boldsymbol{\alpha})\|_2$ (\Cref{loc:simple_bound_on_the_gaussian_noise_error_factor.statement}). In
Section~\ref{loc:body.variable:density_sampling}, we specialize this framework to
Bernoulli selectors by introducing the tight Bernoulli complexity $\gamma(\boldsymbol{\alpha},\boldsymbol{w})$,
proving an RIP bound for Bernoulli sampling (\cref{loc:rip_of_bernoulli_selector_non:uniform_sampling_matrix_on_a_union_of_subspaces.statement}), and combining it with the
results of \Cref{loc:body.general_signal_recovery} to obtain the general recovery and denoising guarantees
(\Cref{loc:cs_with_bernoulli_and_denoising_on_unions_of_subspaces.statement}) for arbitrary
Bernoulli weights $\boldsymbol{w}$. Section~\ref{loc:body.optimized_sampling} then introduces a
simplified upper bound on the Bernoulli complexity, shows that $\boldsymbol{w}^\circ$ uniquely minimizes it, and
identifies $L(\boldsymbol{\alpha},m)$ as the resulting optimized complexity value; it also contributes to the derivation of the
bounds on the noise error term required for the main theorem.
A toy example illustrating the behaviour of the optimized scheme is given in
Section~\ref{loc:body.toy_example}, and numerical comparisons with alternative sampling
strategies appear in Section~\ref{loc:body.numerics}. The proofs of the results stated in
these sections are collected in Section~\ref{loc:body.proofs}.

Much of the exposition follows an analogous structure to that of~\cite{planDenoisingGuaranteesOptimized2025},
and readers familiar with that work may find the associated commentary helpful in interpreting
the results developed here.

\textbf{Notation.}
$\sphere{n}$ is the unit sphere in $\mathbb{R}^n$ or $\mathbb{C}^n$ depending
on context. 
The simplex $\Delta^{n-1}$ is defined as:
\begin{equation*}
\Delta^{n-1} = \left\{ \boldsymbol{p} \in \mathbb{R}^n \mid p_i \geq 0, \sum p_i = 1 \right\}.
\end{equation*}
We let $B_2$ be the $\ell_2$ ball, and $B_2^n$ be the
$\ell_2$ ball of dimension $n$ specifically. We say that a set $\mathcal{T}$ in
a real or complex vector space is a \emph{cone} when $\forall \lambda  \in (0, \infty), \lambda \mathcal{T} = \mathcal{T}$, where $\lambda \mathcal{T} := \{\lambda t | t  \in \mathcal{T}\}$.
The self-difference $\mathcal{V} - \mathcal{V}$ is $\left\{ \boldsymbol{v}_1 - \boldsymbol{v}_2 \mid \boldsymbol{v}_1, \boldsymbol{v}_2 \in \mathcal{V} \right\}$.

Let $\mathbb{R}_+$ be the non-negative real numbers, $\mathbb{R}_{++}$
the strictly positive real numbers, and $\mathbb{N}$ the natural numbers
starting at $1$. For a function $f$, we denote its range by $\range(f)$,
and its restriction to a subset $C$ of its domain by $f|_C$. If
$f$ is a vector and $C$ a subset of its support, then $f|_C$ is the
$|C|$-dimensional vector that is the restriction of $f$ to $C$.
Throughout this paper, we fix the field $\mathbb{K}$ to be either $\mathbb{C}$ or
$\mathbb{R}$.

For a vector $\boldsymbol{u}$, its components are indexed as $u_i$.
We denote by
$\{\boldsymbol{e}_i\}_{i  \in [n]}$ the canonical basis of $\mathbb{R}^n$.
For $\ell \in \mathbb{N}$, the set $[\ell]$ comprises
the integers from $1$ to $\ell$. 
We denote by $\supp \boldsymbol{v}$ the support of $\boldsymbol{v}$, and by
$\boldsymbol{v}^{.2}$ the entry-wise square of $\boldsymbol{v}$.
We say that a vector is \emph{increasing} when $i  \ge j  \implies v_i  \ge v_j$
and \emph{decreasing} when $i  \ge j  \implies v_i  \le v_j$.

For an $m  \times  n$ matrix $A$, we denote its adjoint (the conjugate
transpose) by
$A^*$, its entries by
$A_{i,j}$, and we denote by $\boldsymbol{a}_i$ the conjugate transpose of its rows,
i.e., $A  =  \sum_{i = 1}^m \boldsymbol{e}_i \boldsymbol{a}_i^*$.
The Euclidean norm
of a vector $\boldsymbol{u} \in \measfield^n$ is
$\|\boldsymbol{v}\|_2 := \sqrt{\boldsymbol{v}^* \boldsymbol{v}}$.
The
operator norm of a matrix $A$ is $\|A\|:= \sup_{\boldsymbol{u} \in B_2^n} \|Au\|_2$.
For matrices, given a vector $\boldsymbol{d}  \in \mathbb{R}^n$, we
denote by $\Diag(\boldsymbol{d})$ the $n  \times n$ diagonal
matrix with diagonal entries $\boldsymbol{d}$. The identity matrix in
$\mathbb{R}^m$ is labeled $I_m$.
Projection onto a closed set $\mathcal{T} \subseteq \mathbb{R}^n$ is denoted by $\proj_{\mathcal{T}}$, mapping a vector
$\boldsymbol{x}$ to the element in $\mathcal{T}$ that minimizes the Euclidean
distance, with ties broken by choosing the lexicographically first (meaning
that vectors are ordered by
their first
entry, then second, then third, and so on).

We use $\langle \cdot, \cdot \rangle$ to denote
the inner product in $\mathbb{K}^n$; specifically, the canonical inner product
when $\measfield$ is $\mathbb{R}$, and the complex inner
product $\langle \boldsymbol{u}, \boldsymbol{v}\rangle = \boldsymbol{u}^* \boldsymbol{v}$ when $\measfield$ is $\mathbb{C}$. We also denote by
$\mathcal{R}\langle  \cdot ,  \cdot \rangle$ the real part of the inner
product (which is just the canonical inner product when $\measfield$ is $\mathbb{R}$).

We employ the notation $a \lesssim b$ if $a \leq Cb$ where $C$ is an absolute
constant, potentially different for each instance. 
We denote $X \sim \mathcal{N}(\mu, \sigma^2)$ to be the Gaussian random variable
with mean $\mu$ and variance $\sigma^2$. A random Gaussian vector $g \sim \mathcal{N}(0, I_m)$ is a random vector in $\mathbb{R}^m$ which has i.i.d.
$\mathcal{N}(0, 1)$ Gaussian entries. A complex random Gaussian vector 
$g \sim \mathcal{N}(0, I_m)$, $g  \in \mathbb{C}^m$, is a random
vector in $\mathbb{C}^m$ with entries that have real and imaginary parts individually $\iid \mathcal{N}(0, 1)$.

\section{Main result}
\label{loc:body.main_result}
To set up our main result, we introduce a number of mathematical objects.  We begin with the \emph{local coherence} which quantifies the alignment of a measurement vector with a cone.
\begin{definition}[Local coherence]
\label{loc:local_coherence.statement}
The \emph{local coherence} of a vector $\boldsymbol{\phi} \in \measfield^n$ with respect to a cone $\mathcal{T} \subseteq \field^n$ is defined as
\begin{equation*}
\alpha_{\mathcal{T}}(\boldsymbol\phi) := \sup_{\boldsymbol{x} \in \mathcal{T} \cap B_2} \lvert \boldsymbol\phi^* \boldsymbol{x} \rvert.
\end{equation*}
The \emph{local coherences} of a unitary matrix $F \in \measfield^{n \times n}$ with respect to a cone $\mathcal{T} \subseteq \measfield^n$ is defined as the vector $\boldsymbol{\alpha} \in \mathbb{R}^n_{+}$ with entries $\alpha_j := \alpha_\mathcal{T}(\boldsymbol{f}_j)$, where $\boldsymbol{f}_j$ is the conjugate transpose of the $j^{th}$ row of $F$.
\end{definition}
We note that the local coherence vector can be hard to compute depending on the structure of $\mathcal{T}$.  We give two methods to heuristically approximate it in Section \ref{loc:body.numerics}.

Recall that we consider CS matrices of the form $SF$, where $F$ is a $n  \times n$ unitary matrix (for example, $F$ can be the Fourier
matrix), and $S$ is a \emph{sampling matrix} with Bernoulli selectors, which we now
define.
\begin{definition}[Bernoulli selector sampling matrix]
\label{loc:bernoulli_selector_sampling_matrix.statement}
Let $\boldsymbol{w} \in [0,1]^n \cap m\Delta^{n-1}$ be a \emph{probability weight vector} and $A  \in \mathbb{R}^{n \times n}$ be a random diagonal matrix with entries $A_{i,i} = \sqrt{ \frac{n}{m} } \xi_i$, where $\xi_i \iid \mathrm{Ber}(w_i)$. We define the \emph{Bernoulli sampling matrix} to be the $\tilde{m} \times n$ matrix $S$ obtained by removing the rows of $A$ that are $0$.
\end{definition}
Note that a Bernoulli sampling matrix has a random number of rows $\tilde{m}$ 
satisfying $\mathbb{E} \tilde{m} = m$.
We outline the problem of robust signal recovery in greater detail.
\begin{setup}[{\cite[Setup 2.5]{planDenoisingGuaranteesOptimized2025}}]
\label{loc:setup_bernoulli_signal_recovery.statement}

\textbf{Prior and true signal}

Let $\boldsymbol{x}_{0} \in \field^n$ be a signal, and 
$\mathcal{Q} \subseteq \field^n$ be a prior set
such that $\mathcal{Q}-\mathcal{Q} \subseteq \mathcal{T}$, for $\mathcal{T} \subseteq \mathbb{R}^n$ a union of $M$
subspaces each of dimension at most $\ell$. We think of $\boldsymbol{x}_0$ as being close to $\mathcal{Q}$;  $\boldsymbol{x}^\perp := \boldsymbol{x}_0 - \proj_{\mathcal{Q}} \boldsymbol{x}_0$ quantifies the model mismatch. 
\smallskip 

\textbf{Measurement acquisition}

Let $F  \in \measfield^{n  \times n}$ be a unitary matrix. Suppose that $\boldsymbol{\alpha}$
is the vector of local coherences of $F$ with respect to
$\mathcal{T}$. Let $S$ be a possibly random
\emph{sampling matrix}, and define the measurements
\begin{equation*}
\boldsymbol{b} = SF\boldsymbol{x}_{0} + \boldsymbol{\eta},
\end{equation*}
where the noise is $\boldsymbol{\eta} =  \frac{\sigma \boldsymbol{g}}{\sqrt{m}}$
with 
$\boldsymbol{g} \sim \mathcal{N}(0, I_m)$ being a Gaussian vector in $\measfield^m$. Here,
$\mathbb{E}[\|\boldsymbol{\eta}\|_2^2]$ is $\sigma^2$ when $\mathbb{K}$
is $\mathbb{R}$ and $2 \sigma^2$ when $\mathbb{K}$ is $\mathbb{C}$. Thus, $\sigma$
determines the size of the noise. With this normalization, $\frac{1}{\sigma}$
can be thought of as the Signal-to-Noise Ratio (SNR) up to an absolute constant. 

\smallskip 

\textbf{Signal reconstruction}

Knowing only $\boldsymbol{b}, S$ and $F$, we (approximately) recover the true signal $\boldsymbol{x}_0$ by (approximately) solving the following optimization problem:
\begin{equation}
\label{eq:opt:reconstruct}
\minimize_{\boldsymbol{x} \in \mathcal{Q}} \, \lVert \widetilde{D} SF\boldsymbol{x}-\widetilde{D}\boldsymbol{b} \rVert_{2}^2
\end{equation}
where $\widetilde{D} := \sqrt{\frac{m}{n}}\Diag(S\boldsymbol{d})$ is a diagonal preconditioning
matrix for some $\boldsymbol{d} \in \mathbb{R}^n$.
Note that in terms of $D := \Diag(\boldsymbol{d})$, the preconditioned CS
matrix $\widetilde{D}SF$ can be written as $SDF$. This demonstrates that the
preconditioning is, in fact, an 
element-wise scaling operation applied to individual rows of $F$.
The use of the preconditioner $\widetilde{D}$ is prevalent in the literature
(e.g., see \cite{krahmerStableRobustSampling2014}), and seems to be necessary to obtain the RIP (see \Cref{loc:rip.statement} in \Cref{loc:body.general_signal_recovery}) on the CS matrix when the sampling probabilities are non-uniform. Intuitively, it is chosen so as to ``balance out" the effect of the sampling probabilities on the expected size of the measurement vectors.
We approximately solve the optimization problem \Cref{eq:opt:reconstruct} and obtain an $\hat{\boldsymbol{x}}\in \mathcal{Q}$ such that 
\begin{equation}
\label{eq:opt:recov}
\lVert  \widetilde{D}SF\hat{\boldsymbol{x}}-\widetilde{D}\boldsymbol{b} \rVert_{2}^2 \leq  \min_{\boldsymbol{x} \in \mathcal{Q}}\lVert  \widetilde{D}SF\boldsymbol{x} - \widetilde{D}\boldsymbol{b} \rVert_2^2+\varepsilon
\end{equation}
for some small optimization error $\varepsilon>0$.
\end{setup}

Given the above setup, our objective is to bound the error 
$\|\boldsymbol{x}_0- \hat{\boldsymbol{x}}\|_2$
in terms of the noise level $\sigma$, the optimization error $\varepsilon$,
and the distance of the true signal $\boldsymbol{x}_0$ to the prior
$\mathcal{Q}$ ($\|\boldsymbol{x}^\perp\|$, the approximation error). 

Our setting is identical to~\cite[Setup 2.5]{planDenoisingGuaranteesOptimized2025}, with the only difference that we let $S$ be a sampling matrix with Bernoulli selectors. As will be made apparent, the theoretical challenge (and associated benefit) of Bernoulli selector sampling lies in the possibility of measurement vectors being sampled with probability $1$. We call such measurement vectors \emph{saturated}.

We specify the \emph{optimized} sampling scheme for Bernoulli selectors.
Recall that in this work, a vector $\boldsymbol{\alpha}$ is ``increasing" when $i \le j \implies \alpha_i \le \alpha_j$.
\begin{definition}[Optimized Bernoulli weights]
\label{loc:optimized_bernoulli_weights.statement}
Let $m < n$ and $\boldsymbol{\alpha} \in \mathbb{R}_{++}^n$ be a local coherence vector which, without loss of generality, we assume has increasing entries (otherwise, re-index the vector $\boldsymbol{\alpha}$). Let
\begin{equation*}
R^2(j;\boldsymbol{\alpha}, m) =  \frac{m \|\boldsymbol{\alpha}|_{[j]}\|_2^2}{j-(n-m)},
\end{equation*}
and
\begin{equation}
\label{eq:j}
J = 
\max\left\{  j \in [n]: \frac{m \alpha_j^2}{R^2(j; \boldsymbol{\alpha}, m)} <  1\right\}.
\end{equation}
Then for $L^2(\boldsymbol{\alpha}, m) := R^2(J; \boldsymbol{\alpha}, m)$, the optimized probability weights are
\begin{equation*}
w^\circ_j = \min\left( \frac{m \alpha^2_j}{L^2(\boldsymbol{\alpha}, m)}, 1\right).
\end{equation*}
\end{definition}
The optimized sampling vector is chosen so as to guarantee a small sample complexity in \Cref{loc:optimized_bernoulli_cs_on_union_of_subspaces.statement}.
\begin{remark}
\Cref{loc:optimized_bernoulli_weights.statement} assumes that the entries of
$\boldsymbol{\alpha}$ are ordered increasingly; for any local coherence vector that is not
already ordered, we evaluate $L(\boldsymbol{\alpha}, m)$ by first re-ordering
$\boldsymbol{\alpha}$ increasingly and then applying the same definition.
\end{remark}
\begin{remark}
The quantity $L(\boldsymbol{\alpha},m)$ from Definition~\ref{loc:optimized_bernoulli_weights.statement}
will later be identified (in Section~\ref{loc:body.optimized_sampling}) as the optimized
sample-complexity factor obtained by minimizing a simplified Bernoulli complexity measure $\eta(\boldsymbol{\alpha},\boldsymbol{w})$
over all feasible Bernoulli probability weights with the explicit weights
$\boldsymbol{w}^\circ$ from Definition~\ref{loc:optimized_bernoulli_weights.statement} arising as the
(unique) minimizer. 
\end{remark}

\begin{proposition}[Norm of the optimized probability weights]
\label{loc:norm_of_the_optimized_probability_weights.statement}
In \Cref{loc:optimized_bernoulli_weights.statement},

\begin{enumerate}
\item $J$ is the number of unsaturated measurement vectors.
\item $\sum_{i = 1}^n w_i^\circ =  m$, that is, in expectation $m$ rows are sampled.
\end{enumerate}
\end{proposition}
We defer \hyperlink{loc:norm_of_the_optimized_probability_weights.proof}{the proof} of \Cref{loc:norm_of_the_optimized_probability_weights.statement} to \Cref{loc:body.proofs.properties_of_the_optimized_sampling_scheme}. We are now ready to state our main recovery guarantee for optimized Bernoulli sampling.
\begin{theorem}[Optimized Bernoulli CS on union of subspaces]
\label{loc:optimized_bernoulli_cs_on_union_of_subspaces.statement}
Under~\Cref{loc:setup_bernoulli_signal_recovery.statement}, for some $\delta>0$, with
$L^2(\boldsymbol{\alpha}, m)$ as in~\Cref{loc:optimized_bernoulli_weights.statement},
suppose that
\begin{equation}
\label{eq:samp:comp}
m \gtrsim L^2(\boldsymbol{\alpha}, m) \left( \log \ell + \log M + \log \frac{20}{\delta}\right).
\end{equation}
With the sampling matrix $S$ governed by the optimized probability weights $\boldsymbol{w}^\circ$, and $D = \diag(\boldsymbol{d})$ where $d_i = \sqrt{\frac{m}{nw_i^\circ}}$, the following holds with probability at least $1-\delta$.

For any $\boldsymbol{x}_0 \in \field^n$, with $\varepsilon, \hat{\boldsymbol{x}}, \boldsymbol{x}^{\perp}$ as in~\Cref{loc:setup_bernoulli_signal_recovery.statement}, we have that
\begin{align*}
\lVert \hat{\boldsymbol{x}}- \boldsymbol{x}_0\rVert_2 \leq 9\frac{ \sigma}{\sqrt{ m }} L(\boldsymbol{\alpha}, m) \sqrt{ \min\left(  \frac{5}{4\delta} + \frac{m}{n L^2(\boldsymbol{\alpha}, m)}, \frac{1}{n \min(\boldsymbol{\alpha})^2} \right) } \left( \sqrt{ \ell } + \sqrt{\log M} + \sqrt{ \log \frac{20}{\delta}} \right)&\\
+\lVert \boldsymbol{x}^\perp\rVert + 6\lVert SDF\boldsymbol{x}^\perp\rVert_{2}  + \frac{3}{2}\sqrt{ \varepsilon }.&
\end{align*}
\end{theorem}
The proof of \Cref{loc:optimized_bernoulli_cs_on_union_of_subspaces.statement} will follow from results in the following sections.
\begin{remark}
\label{loc:optimized_bernoulli_cs_on_union_of_subspaces.remark_simplified_result}
\Cref{loc:optimized_bernoulli_cs_on_union_of_subspaces.statement} implies the following. Let $m  \gtrsim L^2(\boldsymbol{\alpha}, m) \left( \log \ell + \log M + 1\right)$, $\boldsymbol{x}_0 \in \mathcal{Q}$. 
We also require, with $\boldsymbol{\alpha}$ re-indexed so that it
is increasing, that $\lVert \boldsymbol{\alpha}|_{\le  n-m+1}\rVert_2^2 \gtrsim \frac{m}{n}$, which is a mild technical condition on the local coherences (see~\Cref{loc:tail_on_the_noise_sensitivity_for_adapted_bernoulli_sampling.remark_bernoulli_paper}).
Then with probability at least $0.99$, any exact minimizer $\hat{\boldsymbol{x}}$ of \Cref{eq:opt:reconstruct} is such that 
\begin{equation*}
\|\hat{\boldsymbol{x}}-\boldsymbol{x}_0\|_2 \lesssim \frac{ \sigma}{\sqrt{ m }} L(\boldsymbol{\alpha}, m)\left( \sqrt{ \ell } + \sqrt{\log M} + 1 \right).
\end{equation*}
\end{remark}
The sample complexity in \Cref{loc:optimized_bernoulli_cs_on_union_of_subspaces.statement} is unconventional in its dependence of the function $L$ on $m$, on the r.h.s. of \Cref{eq:samp:comp}. We leave it in this form because we do not see a way 
to isolate the explicit bound on $m$. Nonetheless, our bound provides meaningful
intuition and improves over the analogous bound for with-replacement sampling.
\begin{proposition}[Upper bound Bernoulli L with local coherences]
\label{loc:upper_bound_bernoulli_l_with_local_coherences.statement}
Given an increasing vector $\boldsymbol{\alpha} \in \mathbb{R}_{++}^n$ and $m  \in \mathbb{N}$, let $L(\boldsymbol{\alpha}, m)$ be as in~\Cref{loc:optimized_bernoulli_weights.statement} (equivalently,~\Cref{loc:optimize_simple_bernoulli_selector_sampling.statement}). Then
\begin{equation*}
\lVert \boldsymbol{\alpha}|_{\leq n-m+1}\rVert_2 \le  L(\boldsymbol{\alpha}, m) \le  \|\boldsymbol{\alpha}\|_2.
\end{equation*}
\end{proposition}
The quantity $\|\boldsymbol{\alpha}\|_2$ is found in place of
$L(\boldsymbol{\alpha}, m)$ in~\cite[Theorem 2.1]{planDenoisingGuaranteesOptimized2025}
for an analogous optimized with-replacement sampling scheme.
Our sample complexity bound for Bernoulli sampling therefore improves
on the with-replacement analogue.
\begin{proposition}[Monotonicity of L in m]
\label{loc:monotonicity_of_l_in_m.statement}
The function $L(\boldsymbol{\alpha},  m)$ is decreasing in $m$.
\end{proposition}
We defer \hyperlink{loc:upper_bound_bernoulli_l_with_local_coherences.proof}{the proof} of \Cref{loc:upper_bound_bernoulli_l_with_local_coherences.statement} and \hyperlink{loc:monotonicity_of_l_in_m.proof}{the proof} of \Cref{loc:monotonicity_of_l_in_m.statement} to \Cref{loc:body.proofs.properties_of_the_optimized_sampling_scheme}.

\Cref{loc:monotonicity_of_l_in_m.statement} suggests that the improvement in the sample complexity bound becomes significant for large
values of $m$. This makes sense intuitively because the quantity $L^2(\boldsymbol{\alpha}, m)$ is the sum of squared local coherences, but
limited to unsaturated entries only (which have smaller local coherences), and then suitably re-normalized. As $m$ grows, more entries
become saturated, and the associated local coherences (which are the larger local coherences) are excluded from the computation of $L^2(\boldsymbol{\alpha}, m)$.
So when $m$ becomes large enough to ``eat up" large local coherences,
the implicit sample complexity bound on $m$ in \Cref{loc:optimized_bernoulli_cs_on_union_of_subspaces.statement} will be easily be satisfied,
because the value of $L^2(\boldsymbol{\alpha}, m)$
will drop. We compute numerically the extent of this drop in \Cref{fig:numerical_complexity_bounds.pdf}.

To compute the explicit bound, recall that our sample complexity bound is of the form
$m \ge L^2(\boldsymbol{\alpha}, m) \Lambda$ for some quantity $\Lambda \in \mathbb{R}_+$. In \Cref{loc:optimized_bernoulli_cs_on_union_of_subspaces.statement},
we have this type of bound with $\Lambda = \left( \log \ell + \log M + \log \frac{20}{\delta}\right)$. 
We derive an associated explicit bound on $m$: with $\Phi(m) := \frac{m}{L^2(\boldsymbol{\alpha}, m)}$,
we have $m \ge \Phi^{-1}(\Lambda)$.
The inverse is well-defined because $L(\boldsymbol{\alpha}, m)$ is
decreasing in $m$ by \Cref{loc:monotonicity_of_l_in_m.statement}, and therefore
$\Phi:\mathbb{R} \to \mathbb{R}$
is a strictly increasing function, making $\Phi$
invertible when restricting its codomain to match its range. We compute
the values of the explicit sample complexity bound for a realistic coherence vector in
\Cref{fig:numerical_complexity_bounds.pdf},
comparing it with the analogous
bound in the with-replacement setting.

\begin{remark}
\label{loc:tail_on_the_noise_sensitivity_for_adapted_bernoulli_sampling.remark_bernoulli_paper}
The term $\frac{1}{12}\frac{m}{n L^2}$ in \Cref{loc:optimized_bernoulli_cs_on_union_of_subspaces.statement} is usually small. Indeed, from \Cref{loc:upper_bound_bernoulli_l_with_local_coherences.statement} it follows that after re-ordering $\boldsymbol{\alpha}$ to be increasing, $\frac{m}{n L^2} \leq \frac{m}{n} \frac{1}{\lVert \boldsymbol{\alpha}|_{\le  n-m+1}\rVert_2^2}$. This term becomes significant only when
\begin{equation*}
\lVert \boldsymbol{\alpha}|_{\le  n-m+1}\rVert_2^2 \lesssim \frac{m}{n}.
\end{equation*}
All but the top $m$ local coherences would have to be on the order of $\frac{\sqrt{m}}{n}$ (assuming that $m \ll n$).
The expected coherence of a randomly rotated vector with a $k$-dimensional subspace is about $\sqrt{\frac{k}{n}}$, so the typical size of the truncated local coherence vector is on the order of $k$. Therefore, for the additional term to be significant, the local coherences would
have to be atypically small.
\end{remark}
\begin{figure}[!t]
\centering
\includegraphics[width=\textwidth,alt={Two log-scale line plots for a fixed local coherence vector from the flower-data discrete Fourier transform prior set. In the left panel, L squared of alpha and m decreases as the ratio m over n increases, while the squared ell two norm of alpha stays constant. In the right panel, the sample-complexity bound increases with capital Lambda for both Bernoulli and with-replacement sampling, and the Bernoulli bound stays lower throughout.}]{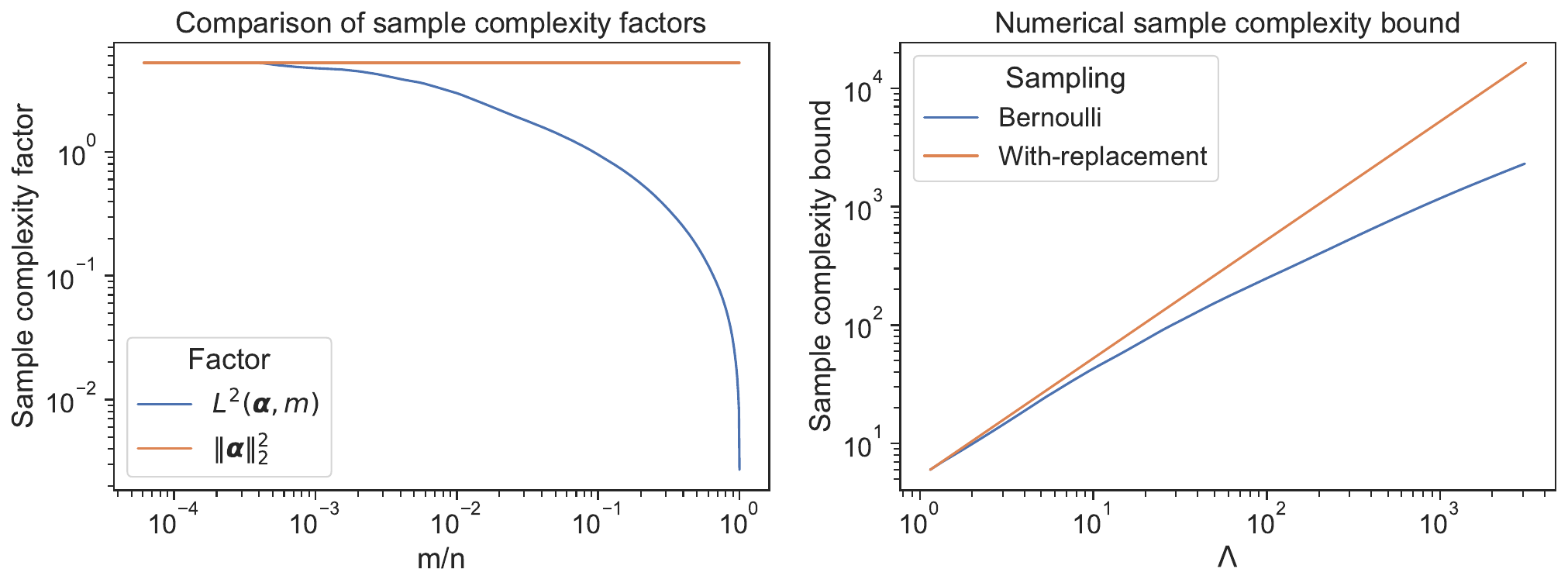}
\caption{
    We consider a fixed coherence vector $\boldsymbol{\alpha}$ across all experiments,
    which is the local coherences of the DFT matrix on the flower dataset~\cite{nilsbackAutomatedFlowerClassification2008} as prior set ($n = 16384$). 
    The right plot compares numerically the bound on $m$ induced by a bound of the form $m \ge L^2(\boldsymbol{\alpha}, m) \Lambda$ (see the above discussion).
}\label{fig:numerical_complexity_bounds.pdf}
\end{figure}

\section{General signal recovery}
\label{loc:body.general_signal_recovery}

In this section we develop a general recovery framework that applies to any sensing
matrix satisfying a suitable RIP condition. Working in the setting of \Cref{loc:setup_bernoulli_signal_recovery.statement}, we derive a
signal recovery guarantee under RIP assumptions on the preconditioned CS matrix $S D F$.
This abstract result forms the backbone of our analysis: Section~\ref{loc:body.variable:density_sampling}
will verify the required RIP condition specifically for Bernoulli selectors, and define
the complexity quantity $\gamma(\boldsymbol{\alpha},\boldsymbol{w})$ under which these bounds hold.

To control the noise term appearing in our recovery bound, we introduce a simple
truncation operator that limits the magnitude of a vector to unit scale by truncating its support to sufficiently many entries (using the coordinate order).

\begin{definition}[Unit truncation]
\label{loc:unit_truncation.statement}
Given some $\boldsymbol{v} \in \mathbb{R}^n$, let
\begin{equation*}
I =\min \left\{ \bar{I} \in  [n] \middle| \|v|_{[\bar{I}]}\|_2 \ge 1 \right\}.
\end{equation*}
Then define the unit truncation operator $\trunc: \measfield^n  \to  \measfield^n$ to have
entries
\begin{equation*}
\trunc(\boldsymbol{v})_i := \begin{cases}
v_i & i < I, \\
\sqrt{1-\|\boldsymbol{v}|_{[I-1]}\|_2^2} & i  = I,  \\
0 & i> I.
\end{cases}
\end{equation*}
\end{definition}
\begin{remark}
The truncation operator is used solely to bound the contribution of the noise term
in our analysis; it does not affect the structure of the recovery argument.
\end{remark}

Next we recall the general version of the celebrated \emph{restricted isometry property (RIP)}.
\begin{definition}[Restricted Isometry Property]
\label{loc:rip.statement}
Let $\mathcal{T} \subseteq \field^n$ be a cone and $A \in \measfield^{m \times n}$ a matrix. We say that $A$ satisfies the \emph{Restricted Isometry Property} (RIP) when
\begin{equation*}
\sup_{u \in \mathcal{T} \cap \mathcal{S}^{n-1}}\lvert \lVert Au\rVert_2 - 1\rvert \leq \frac{1}{3}.
\end{equation*}
\end{definition}

The following lemma provides a general signal recovery guarantee under an RIP
assumption. It was originally stated in  \cite[Theorem 3.3]{planDenoisingGuaranteesOptimized2025}, 
we restate it here for completeness.

\begin{lemma}[Signal recovery with the RIP;{\cite[Theorem 3.3]{planDenoisingGuaranteesOptimized2025}}]
\label{loc:signal_recovery_with_subsampled_unitary_matrix_with_gaussian_noise_on_union_of_subspaces.bernoulli} Under~\Cref{loc:setup_bernoulli_signal_recovery.statement}, let $S$ be some fixed sampling matrix and suppose that $SDF$ satisfies the \hyperref[loc:rip.statement]{RIP} on
$\mathcal{T}$. Then for $t>0$, the following holds with probability at least
$1-2\exp(-t^2)$. 

For any $\boldsymbol{x}_0 \in \field^n$, with
$\varepsilon, \hat{\boldsymbol{x}}, \boldsymbol{x}^{\perp}$ as
in~\Cref{loc:setup_bernoulli_signal_recovery.statement},
we have that
\begin{align*}
\lVert \hat{\boldsymbol{x}}- \boldsymbol{x}_0\rVert_2 \leq 9\frac{ \sigma}{\sqrt{ m }} \|\widetilde{D}\trunc(SD \boldsymbol{\alpha})\|_2 \left( \sqrt{ \ell } + \sqrt{\log M} + t\right)&\\
+\lVert \boldsymbol{x}^\perp\rVert + 6\lVert SDF\boldsymbol{x}^\perp\rVert_{2}  + \frac{3}{2}\sqrt{ \varepsilon }.
\end{align*}
\end{lemma}

To apply the recovery bound of \Cref{loc:signal_recovery_with_subsampled_unitary_matrix_with_gaussian_noise_on_union_of_subspaces.bernoulli}, we require estimates on the noise factor
$\|\tilde D\,\trunc(SD\boldsymbol{\alpha})\|_2$. The next result gives general upper bounds for this term, which will later be specialized to the Bernoulli setting in Section~\ref{loc:body.variable:density_sampling}.

\begin{proposition}[Bounds on the noise error]
\label{loc:simple_bound_on_the_gaussian_noise_error_factor.statement}
With any $\boldsymbol{d} \in \mathbb{R}^n_{++}$, $\boldsymbol{\alpha} \in \mathbb{R}^n_{++}$, and $S$ a fixed $m  \times  n$ sampling matrix such that $S \boldsymbol{d}$ is decreasing, with $D  =  \Diag(\boldsymbol{d})$ and $\widetilde{D} = \Diag( \sqrt{\frac{m}{n}} S \boldsymbol{d})$ we have that
\begin{equation}
\label{eq:first:trunc}
\|\widetilde{D}\trunc(SD \boldsymbol{\alpha})\|_2 \leq  \max(S \boldsymbol{d}) \leq  \max(\boldsymbol{d}).
\end{equation}
Furthermore, with
\begin{equation*}
I = |\supp\trunc(SD \boldsymbol{\alpha})|
\end{equation*}
\begin{equation}
\label{eq:second:trunc}
\|\widetilde{D}\trunc(SD \boldsymbol{\alpha})\|_2  \le \|(SD^2 \boldsymbol{\alpha})|_{[I]}\|_2 \leq \|SD^2 \boldsymbol{\alpha}\|_2.
\end{equation}
\end{proposition}
\begin{proof}[\hypertarget{loc:simple_bound_on_the_gaussian_noise_error_factor.proof}Proof of \Cref{loc:simple_bound_on_the_gaussian_noise_error_factor.statement}]

We write
\begin{equation*}
\|\widetilde{D}\trunc(SD \boldsymbol{\alpha})\|_2  = \sqrt{\sum_{i = 1}^I (S \boldsymbol{d})_i^2 \trunc(SD \boldsymbol{\alpha})_i^2}.
\end{equation*}
Under the square root, we find a convex combination of the entries of $S \boldsymbol{d}$ with convex coefficients $\trunc(SD \boldsymbol{\alpha})^{.2}$. The first bound in \Cref{eq:first:trunc} follows from bounding the convex combination by the size of the maximal element.

To see that \Cref{eq:second:trunc} holds, it suffices to check that $(SD \boldsymbol{\alpha})|_{[I]}$ dominates $\trunc(SD \boldsymbol{\alpha})$ entry-wise.
\end{proof}

The recovery lemma and noise bounds established in this section do hold for arbitrary 
sampling matrices. In the next section, we specialize these results to the
Bernoulli setting by introducing the complexity measure $\gamma(\boldsymbol{\alpha},\boldsymbol{w})$ and
establishing conditions under which a Bernoulli CS matrix has
the RIP (with high probability), conditions under which \Cref{loc:signal_recovery_with_subsampled_unitary_matrix_with_gaussian_noise_on_union_of_subspaces.bernoulli} then
applies.

\section{Variable-density sampling}
\label{loc:body.variable:density_sampling}
In this section we specialize the general recovery framework developed in
\Cref{loc:body.general_signal_recovery} to the Bernoulli setting. Given a Bernoulli selector matrix $S$ with
weights $\boldsymbol{w}$, we introduce a tight Bernoulli complexity measure $\gamma(\boldsymbol{\alpha},\boldsymbol{w})$
that captures the interaction between the sampling distribution and the local
coherence vector. This quantity is fundamental in establishing RIP conditions for
Bernoulli sampling, and serves as the bridge between the general recovery bound of
\Cref{loc:signal_recovery_with_subsampled_unitary_matrix_with_gaussian_noise_on_union_of_subspaces.bernoulli} and the optimized sampling scheme developed in \Cref{loc:body.optimized_sampling}.

\begin{definition}[Tight Bernoulli complexity]
\label{loc:tight_bernoulli_complexity.statement}
Let $\boldsymbol{\alpha} \in \mathbb{R}^n_{++}$ be a vector of local coherences and let $\boldsymbol{w} \in (0, 1]^n \cap m \Delta^{n-1}$ for $m < n$. We define
\begin{equation*}
\gamma(\boldsymbol{\alpha}, \boldsymbol{w}) := \max_{j \in [n]} \alpha_j \sqrt{ m } \max\left( \sqrt{ \frac{1-w_j}{w_j} }, 1 \right) \indicator_{\{w_j < 1\}}.
\end{equation*}
\end{definition}

The next result shows that the preconditioned sensing matrix $\tilde D S F$ (equivalently,
$SDF$) satisfies the RIP under conditions controlled by the Bernoulli complexity
$\gamma(\boldsymbol{\alpha},\boldsymbol{w})$. This provides the main technical link between the sampling distribution
$\boldsymbol{w}$ and the RIP-based recovery bound of \Cref{loc:signal_recovery_with_subsampled_unitary_matrix_with_gaussian_noise_on_union_of_subspaces.bernoulli}.

\begin{lemma}[RIP of CS matrix on a union of subspaces]
\label{loc:rip_of_bernoulli_selector_non:uniform_sampling_matrix_on_a_union_of_subspaces.statement}
Considering~\Cref{loc:setup_bernoulli_signal_recovery.statement}, let $S$ be a sampling matrix
with probability weights $\boldsymbol{w}  \in m \Delta^{n-1} \cap (0, 1]$. 
Let $D = \Diag(\boldsymbol{d})$ where $d_i = \sqrt{\frac{m}{n w_i}}$.
For $t>0$, and $\gamma$ as defined in~\Cref{loc:tight_bernoulli_complexity.statement}, suppose that
\begin{equation*}
m \gtrsim \gamma^2(\boldsymbol{\alpha}, \boldsymbol{w}) \left( \log \ell + \log M+ t^2\right).
\end{equation*}
Then for $S$ a sampling matrix with probability weights $\boldsymbol{w}$, $SDF$ has the \hyperref[loc:rip.statement]{RIP} on $\mathcal{T}$ with probability at least $1-2\exp(-t^2)$.
\end{lemma}
The proof can be found in \Cref{loc:body.proofs.rip_of_preconditioned_subsampled_unitary_matrices}.
Combining \Cref{loc:rip_of_bernoulli_selector_non:uniform_sampling_matrix_on_a_union_of_subspaces.statement} with \Cref{loc:signal_recovery_with_subsampled_unitary_matrix_with_gaussian_noise_on_union_of_subspaces.bernoulli}
yields the following compressed sensing recovery and denoising guarantees for variable-density sampling.
\begin{theorem}[CS with Bernoulli and denoising on unions of subspaces]
\label{loc:cs_with_bernoulli_and_denoising_on_unions_of_subspaces.statement}
Under~\Cref{loc:setup_bernoulli_signal_recovery.statement},
with $\gamma$ the complexity function given by~\Cref{loc:simplified_bernoulli_complexity.statement}
and $\delta > 0$, suppose that
\begin{equation*}
m \gtrsim \gamma^2(\boldsymbol{\alpha}, \boldsymbol{w}) \left( \log \ell + \log M + \log \frac{4}{\delta} \right).
\end{equation*}
Sample the sampling matrix $S$ with probability weights $\boldsymbol{w}$. Let $D = \Diag(\boldsymbol{d})$ and $\widetilde{D} = \sqrt{\frac{m}{n}} \Diag(S \boldsymbol{d})$, where $d_i  =  \sqrt{\frac{m}{nw_i}}$.
Then, the following holds with probability at least $1-\delta$.

For any $\boldsymbol{x}_0  \in \field^n$, with $\hat{\boldsymbol{x}}, \boldsymbol{x}^{\perp}, \varepsilon$ as in~\Cref{loc:setup_bernoulli_signal_recovery.statement}, we have that
\begin{align*}
\lVert \hat{\boldsymbol{x}}- \boldsymbol{x}_0\rVert_2 \leq 9\frac{\sigma}{\sqrt{ m }}  \|\widetilde{D}\trunc(SD \boldsymbol{\alpha})\|_2  \left( \sqrt{ \ell } + \sqrt{\log M} + \sqrt{\log \frac{4}{\delta}}\right)&\\
+\lVert \boldsymbol{x}^\perp\rVert + 6\lVert SDF\boldsymbol{x}^\perp\rVert_{2}  + \frac{3}{2}\sqrt{ \varepsilon }.
\end{align*}
\end{theorem}
See \hyperlink{loc:cs_with_bernoulli_and_denoising_on_unions_of_subspaces.proof}{the proof} in \Cref{loc:body.proofs.union_bounds}.

\Cref{loc:cs_with_bernoulli_and_denoising_on_unions_of_subspaces.statement} applies to any Bernoulli sampling distribution and expresses the recovery
guarantee explicitly in terms of the complexity $\gamma(\boldsymbol{\alpha},\boldsymbol{w})$. In the next
section, we introduce a simplified complexity measure $\eta(\boldsymbol{\alpha},\boldsymbol{w})$ and show that
the optimized weights $\boldsymbol{w}^\circ$ minimize this quantity, thereby identifying the
optimized complexity value $L(\boldsymbol{\alpha},m)$ appearing in the main theorem.

\section{Optimized sampling}
\label{loc:body.optimized_sampling}
Here, we define a simplified complexity $\eta(\boldsymbol{\alpha},\boldsymbol{w})$, optimize it in closed form, and use the resulting weights to control the noise-sensitivity term.
\begin{definition}[Simpler complexity function]
\label{loc:simplified_bernoulli_complexity.statement}
Let $\boldsymbol{\alpha} \in \mathbb{R}^n_{++}$ be a vector of local coherences and $\boldsymbol{w} \in (0, 1]^n \cap m \Delta^{n-1}$ for $m < n$. We define
\begin{equation*}
\eta(\boldsymbol{\alpha}, \boldsymbol{w}) := \max_{j \in [n]} \alpha_j \sqrt{ \frac{m}{w_j} } \indicator_{\{w_j < 1\}}.
\end{equation*}
\end{definition}
A direct comparison shows that
$\gamma(\boldsymbol{\alpha},\boldsymbol{w})\le \eta(\boldsymbol{\alpha},\boldsymbol{w})$
for all feasible weights $\boldsymbol{w}\in(0,1]^n\cap m\Delta_{n-1}$. In particular, any upper bound expressed in terms of $\eta(\boldsymbol{\alpha},\boldsymbol{w})$ can be used to control the RIP condition from Lemma 4.2 via $\gamma(\boldsymbol{\alpha},\boldsymbol{w})$. 
The vector of optimized probability weights is the minimizer of $\eta$ over the feasible Bernoulli weights. 
\begin{proposition}[Optimize simple Bernoulli selector sampling]
\label{loc:optimize_simple_bernoulli_selector_sampling.statement}
For a fixed vector $\boldsymbol{\alpha} \in \mathbb{R}_{++}$ and the function $\eta$ as in~\Cref{loc:simplified_bernoulli_complexity.statement}, let $\boldsymbol{w}^{\circ}$ be the optimized probability vector from~\Cref{loc:optimized_bernoulli_weights.statement}. Then
\begin{equation*}
\min_{\boldsymbol{w} \in (0,1]^n \cap m\Delta^{n-1}}\eta(\boldsymbol{\alpha}, \boldsymbol{w})  = \eta(\boldsymbol{\alpha}, \boldsymbol{w}^{\circ})  = L(\boldsymbol{\alpha}, m),
\end{equation*}
where $\boldsymbol{w}^{\circ}$ is the \emph{unique} minimizer.
\end{proposition}
We defer \hyperlink{loc:optimize_simple_bernoulli_selector_sampling.proof}{the proof} to \Cref{loc:body.proofs.optimizing_the_sampling_scheme}.

For the optimized sampling scheme, there is a simple probabilistic upper-bound
on the signal recovery noise error term.

\begin{lemma}[Tail bound for the noise sensitivity under optimized Bernoulli sampling]
\label{loc:tail_on_the_noise_sensitivity_for_adapted_bernoulli_sampling.statement}
Let $\boldsymbol{\boldsymbol{\alpha}} \in \mathbb{R}^n_{++}$ and $m \in \mathbb{N}$, and let $\boldsymbol{w}^\circ$, $J$, and $L(\boldsymbol{\alpha},m)$ be as in \Cref{loc:optimized_bernoulli_weights.statement}. Let $D  = \Diag(\boldsymbol{d})$ and $\widetilde{D}= \Diag\left( \sqrt{ \frac{m}{n} }S\boldsymbol{d} \right)$, where $d_i  = \sqrt{ \frac{m}{n w^{\circ}_i}}$. Let $S$ be the Bernoulli sampling matrix with probability weights $\boldsymbol{w}^{\circ}$, with rows re-ordered so that $S \boldsymbol{d}$ is decreasing. Then for any $t>0$,
\begin{equation*}
\|\widetilde{D}\trunc(SD \boldsymbol{\alpha})\|_2 \leq  L(\boldsymbol{\alpha}, m) \sqrt{ \min\left(  \frac{1}{t} + \frac{m}{n L^2(\boldsymbol{\alpha}, m)}, \frac{1}{n \min(\boldsymbol{\alpha})^2} \right) }
\end{equation*}
with probability at least $1-t$,  where $\mathbb{T}$ is the unit truncation operator from \Cref{loc:unit_truncation.statement}.
\end{lemma}

This completes the last ingredient needed for the optimized recovery guarantee stated in \Cref{loc:optimized_bernoulli_cs_on_union_of_subspaces.statement}. The proof follows by combining \Cref{loc:cs_with_bernoulli_and_denoising_on_unions_of_subspaces.statement} with \Cref{loc:tail_on_the_noise_sensitivity_for_adapted_bernoulli_sampling.statement} via a union bound (see \hyperlink{loc:optimized_bernoulli_cs_on_union_of_subspaces.proof}{the proof} in \Cref{loc:body.proofs.union_bounds}).

Next, we examine the difference between Bernoulli sampling and without-replacement
sampling. In the latter, one repeatedly draws measurement vectors according to a fixed distribution (as in with-replacement sampling), but rejects any draw that repeats a previously selected vector, continuing until 
$m$ distinct measurement vectors are obtained. Equivalently, this procedure can be viewed as sequential sampling without replacement in which, after each accepted draw, the distribution is renormalized over the remaining (unselected) vectors. For completeness, we prove this equivalence in \Cref{app:equivalence_sampling}. Without-replacement sampling 
avoids the pitfall of with-replacement sampling, which often
re-samples measurement vectors with diminishing returns. This 
drawback becomes pronounced for optimized sampling schemes, as we show in
\Cref{loc:body.numerics}. Optimized without-replacement 
sampling with theoretical guarantees was introduced in \cite{hoppeSamplingStrategiesCompressive2023}.

In the next section we outline a toy example consisting of
a simple local coherence vector, for which we show that
optimized Bernoulli sampling significantly outperforms optimized without-replacement in sampling complexity.
\section{Toy example}
\label{loc:body.toy_example}
Consider a prior set $\mathcal{Q}  \subseteq \mathbb{R}^n$ to be the
union of $\{\boldsymbol{e}_1\}$ and of a subspace $\mathcal{U}$ consisting of vectors supported on $\{2,  ...,  n\}$, and which is maximally incoherent with
the canonical basis. 

We now define $\mathcal{U}$ explicitly: let $A  \in \mathbb{R}^{n-1 \times n-1}$ be the discrete Fourier
transform matrix in $\mathbb{R}^{n-1}$, and let $M  \in \mathbb{R}^{n \times n-1}$ be the matrix $A$
padded with zeros in its first row, in the sense that $M_{1, j} = 0$, and $M_{i, j} = A_{i-1,  j}$ for
$i  \in \{2,  \ldots,  n\}$ and $j  \in [n-1]$. We define $\mathcal{U}$ to be the span of the first $k-1$ columns of $M$.

Next, take the identity matrix $I$ as the unitary measurement matrix $F$, so that the CS matrix $SF$ reduces to 
$S$. 

For this choice of $\mathcal{Q}$ and $F$, one can check that $S$ has the RIP if and only if the following hold:
\begin{enumerate}
\item The measurement vector $\boldsymbol{e}_1$ is sampled,
\item Measurement vectors corresponding to at least $k-1$ distinct indices out of $\{2,  \ldots, n\}$ are sampled.
\end{enumerate}

We show that in this simple setting, optimized without-replacement sampling fails
to sample the first measurement vector with constant probability as $n \to 
\infty$, even for large
$k$. In contrast, under the same setup, the probability of failure
for the optimized Bernoulli selector sampling scheme vanishes.

Since $\mathcal{U}$ is a fully-incoherent $(k-1)$-dimensional subspace in the $(n-1)$-dimensional subspace of vectors supported on $\{2, \ldots, n\}$, it follows from \cite[Proposition 3.1]{berkCoherenceParameterCharacterizing2022} that the local coherence vector of $I$ with respect to $\mathcal{Q}$ is 
$$\boldsymbol{\alpha} = \left\{1, \sqrt{\frac{k-1}{n-1}},  \ldots,  \sqrt{\frac{k-1}{n-1}}\right\}.$$
Then the optimized probability vector defined in~\cite[Definition 2.6]{planDenoisingGuaranteesOptimized2025} is 
$$\boldsymbol{p}^* =  \left\{\frac{1}{k}, \frac{k-1}{k(n-1)}, \ldots, \frac{k-1}{k(n-1)}\right\}.$$
Sampling sequentially without replacement according to $\boldsymbol{p}^*$, the probability that the first measurement vector fails to be sampled on the next draw after $\ell$ draws is $$1- \frac{1}{k} \frac{1}{1- \ell \frac{k-1}{k(n-1)}},$$ which, when $m << n$, is approximately $\approx 1- \frac{1}{k}$.
Therefore, as $n \to \infty$ the probability of missing the first measurement $\ell$ times in a row converges to $\left(1- \frac{1}{k}\right)^m$.
With $m = 2k$,
\begin{equation*}
\lim_{k  \to  \infty}\left(1- \frac{1}{k}\right)^{2k}  = \exp(-2),
\end{equation*}
and for any $k  \ge 4$, $\left(1- \frac{1}{k}\right)^{2k} \ge 0.1$ (note that this limit is to be taken only after taking $n  \to  \infty$). So even for arbitrarily large values of $m$ and $k$, so long as $m << n$, there is a fixed positive probability that the RIP fails to be satisfied for the optimized without-replacement sampling scheme. 

On the other hand, optimized Bernoulli selector sampling will have the RIP with vanishing probability of failure as $m$ gets larger. Indeed, for $m \ge k$, the optimized probability weights are $\boldsymbol{w}^* = \{1, \frac{m-1}{n-1}, \ldots, \frac{m-1}{n-1}\}$, whereby the first measurement is always sampled. There is still a positive probability that the RIP will not be satisfied due to the total number of sampled measurements $\tilde{m}$ potentially falling below the required $k$. The probability of this event vanishes for $m=2k$ as $k  \to  \infty$ because $\tilde{m}$ concentrates around $m$ with sub-Gaussian norm of order $\sqrt{m}$ (see, e.g., \cite{rudelsonSparseReconstructionFourier2008}).
\section{Numerics}
\label{loc:body.numerics}
We performed a range of experiments measuring the performance of the
optimized Bernoulli sampling scheme in signal recovery problems, recovering
images from the CELEBA dataset~\cite{liuDeepLearningFace2015}, which are of
size $128\times 112\times 3$, with total dimension $n = 43008$.
For the unitary part of the CS matrix $F$ we use the
two-dimensional DFT, applied channel-wise on the three color channels of the
images.
We run every experiment on two different prior sets: sparse signals in
the Haar wavelet basis with three levels, and signals generated by
a feedforward neural network. 
We briefly describe how we
compute an (approximate) local coherence vector, and then (approximately) solve the recovery optimization
problem of \Cref{eq:opt:reconstruct} in each setting. 
\subsection{Sparsity-based prior}
\label{loc:body.numerics.sparsity:based_prior}
We consider images from the CELEBA dataset that are truncated to their largest
500 coefficients (1\%) in the Haar wavelet basis of three levels.
The coherence vector $\boldsymbol{\alpha}$ is computed with respect to the Haar basis vectors. While this does not exactly match the assumptions of our theory---which require coherence relative to the set of $k$-sparse vectors rather than $1$-sparse vectors---it serves as a tractable heuristic consistent with standard notions of local coherence in sparse signal recovery.
For the Haar basis $\Phi$, a measurement vector $\boldsymbol{f}_i$ has a local coherence $\max_{\boldsymbol{\phi} \in \Phi} \lvert \boldsymbol{f}_i^* \boldsymbol{\phi} \rvert$. The computation is sped up significantly by considering a representative subset of the Haar basis vectors, including one Haar wavelet basis vector for each block of Haar wavelets, utilizing the 
fact that translations of the support of wavelet vectors does not affect
their local coherences in the Fourier basis. This method contrasts with that of~\cite{krahmerStableRobustSampling2014}, where they derive an upper-bound on
the local coherences.

To recover the signal, we utilize basis pursuit denoising with the SPGL1
solver. From the solution, we take as support the indices supporting the most mass.
Restricting to this support, we then solve a least-squares problem, which yields our recovered
signal $\hat{\boldsymbol{x}}$.
\subsection{Generative prior}
\label{loc:body.numerics.generative_prior}
We train a convolutional VAE on the CELEBA
dataset, using the decoder part of the VAE as our feedforward neural network
$G:\mathbb{R}^k \to \mathbb{R}^n$ with $k = 200$ (0.5\%). We compute the coherence vector with
respect to a
random sample in the range of $G$, in the manner
described in
\cite{cardenasCS4MLGeneralFramework2023, berkModeladaptedFourierSampling2023}.
We then attempt to recover a randomly sampled true signal $\boldsymbol{x}_0 =  G(\boldsymbol{z})$
for $\boldsymbol{z} \sim \mathcal{N}(0, I_k)$. We 
sample $S$,  
compute the measurements $\boldsymbol{b} = SF \boldsymbol{x}_0$, and solve the optimization problem
\begin{equation*}
\minimize_{\boldsymbol{z} \in \mathbb{R}^k} \, \lVert \widetilde{D} SF G(\boldsymbol{z}) -\widetilde{D}\boldsymbol{b} \rVert_{2}^2.
\end{equation*}
To solve this program we use the LBFGS~\cite{liuLimitedMemoryBFGS1989} algorithm with 4 steps, running 5 attempts
with random restarts, and using the result of the attempt with the smallest objective value.
We find an approximate minimizer $\hat{\boldsymbol{z}} \in \mathbb{R}^k$,
and a recovered signal $\hat{\boldsymbol{x}} :=  G(\hat{\boldsymbol{z}})$.
\subsection{Experiments}
\label{loc:body.numerics.experiments}
In all experiments, we use an adjusted version of Bernoulli selector
sampling: with the optimized Bernoulli probability weights, we resample $S$
until exactly $m$ measurement vectors are
selected. This is numerically trivial, and is justified by the fact that the
costly operation in the compressed sensing model is not in sampling the measurement vectors,
but in actually measuring the true signal
(computing $\boldsymbol{b} := SF \boldsymbol{x}_0$).
The sampling distribution in our experiments is therefore optimized Bernoulli selectors
conditioned on the total number of sampled measurement vectors matching
the expected number of sampled measurements $m$. 
In all experiments we compute the relative error $\|\hat{\boldsymbol{x}} - \boldsymbol{x}_0\|_2/\|\boldsymbol{x}_0\|_2$.
\begin{figure}[!t]
\centering
\includegraphics[width=\textwidth,alt={Two log-scale plots of relative reconstruction error versus the number of measurements m, with generative signals on the left and sparse signals on the right. Three Bernoulli sampling schemes are compared: homogeneous, deterministic, and optimized. In both panels, optimized sampling reaches low error first, deterministic is close behind, and homogeneous sampling improves much later.}]{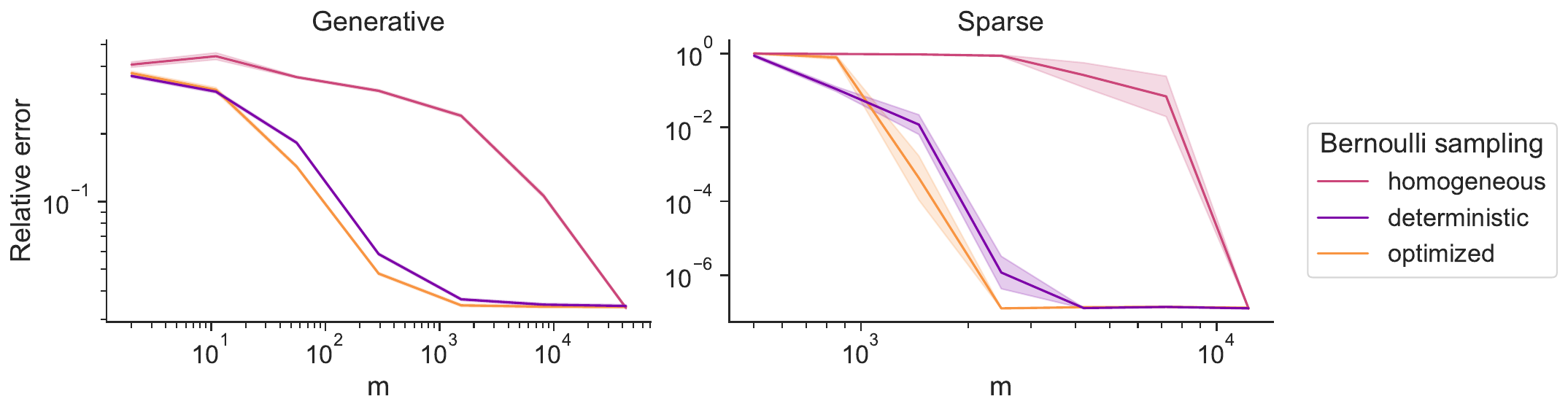}
\caption{The generative plot has 200 experiments for each data point, and
    the sparse plot has 20.  Sparsity level is $k = 500$ (1\%) and the code dimension of the generative model is $k = 200$ (0.5\%). We display a line for geometric mean and a band for the geometric standard error (the uncertainty of the geometric mean estimator).
}\label{fig:partially_det_advantage.pdf}
\end{figure}

Optimized sampling is a compromise between two canonical approaches: using the set of measurement vectors with the highest coherence, 
and using a uniformly random subset of the measurement vectors which diversifies between all possible
measurement vectors. Uniform randomization was the main object of study in compressed sensing. In \Cref{fig:partially_det_advantage.pdf} we see that when comparing
optimized sampling with  alternatives, optimized sampling generally performs the best. 
We note that for very sparse signals ($k < 200$), we found that
deterministic sampling outperforms optimized sampling.
\begin{figure}[!t]
\centering
\includegraphics[width=\textwidth,alt={Two log-scale plots of relative reconstruction error versus the number of measurements m, with generative signals on the left and sparse signals on the right. Three optimized sampling implementations are compared: with-replacement, without-replacement, and Bernoulli. Bernoulli and without-replacement perform similarly and reach low error earlier than with-replacement, especially in the sparse setting.}]{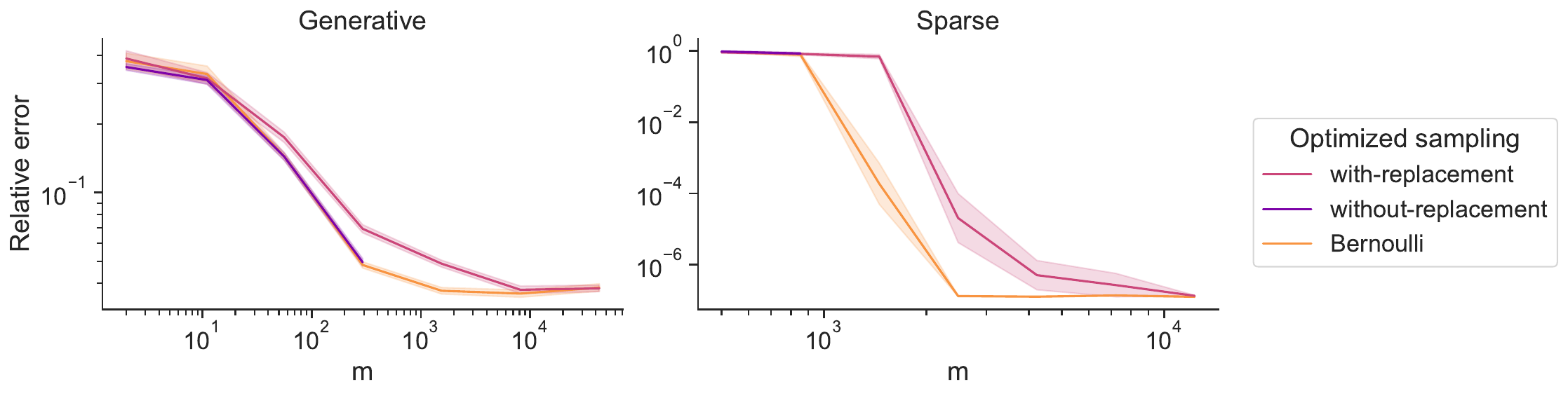}
\caption{Comparison between optimized sampling distribution. 
    The generative plot has 200 experiments for each data point, and
    the sparse plot has 20. We display a line for geometric mean and a band for the geometric standard error (the uncertainty of the geometric mean estimator).  Sparsity level is $k = 500$ (1\%) and the code dimension of the generative model is $k = 200$ (0.5\%). }\label{fig:main_ber_comparison.pdf}
\end{figure}

In \Cref{fig:main_ber_comparison.pdf} we compare Bernoulli selector sampling to with-replacement sampling,
a sampling distribution common in the literature (see, e.g., \cite{planDenoisingGuaranteesOptimized2025}).
We find that
optimized Bernoulli selector sampling outperforms with-replacement sampling,
which can be explained by the tendency of optimized with-replacement sampling 
to include the same high-coherence measurement vectors repeatedly in the
CS matrix. Redundant measurement vectors do not provide any additional information
about the true signal, other than in their denoising properties, so the CS
matrix will effectively contain a smaller number of measurement vectors. 
This shortcoming was addressed in~\cite{hoppeSamplingStrategiesCompressive2023},
where the authors provided theoretical guarantees for an optimized
without-replacement sampling scheme, which we include in \Cref{fig:main_ber_comparison.pdf}.
The preconditioner introduced in~\cite{hoppeSamplingStrategiesCompressive2023}
simulates sampling vectors with-replacement until there are the desired
number of distinct measurement vectors. Then, they include each distinct
measurement vector only once in the CS matrix, removing the measurement
cost of duplicate measurement vectors. They incorporate the count of
redundant sampling for each measurement vector in the preconditioner,
so that theoretical guarantees from with-replacement sampling carry
over to this without-replacement sampling counterpart.
The main drawback of this method is that to compute this ``empirical" preconditioner,
it is necessary to use rejection sampling, with no recourse to conditional
probability distributions, because the preconditioner is formulated with the count of rejections for each measurement vector. 

Sampling this way proved computationally infeasible for large numbers of measurements and steeply decaying local coherences. Indeed, once the measurement vectors with high local coherence
 have already been
sampled, the remaining unsampled measurement vectors have vanishingly small 
combined probability mass relative to their sampled counterparts.
The number of attempts required to sample the next
novel measurement vector becomes too large for even a fast loop on a computer.
We stopped evaluating this method in \Cref{fig:main_ber_comparison.pdf} when
the computational cost grew too large, even with reasonably optimized code (whenever the estimated number of sampling attempts reached $10^7$).
\begin{figure}[ht]
\centering
\includegraphics[width=\textwidth,alt={Two log-scale plots of relative reconstruction error versus the number of measurements m, with generative signals on the left and sparse signals on the right. Three methods are compared: optimized without-replacement sampling with heuristic preconditioning, optimized without-replacement sampling with empirical preconditioning, and optimized Bernoulli sampling with Bernoulli preconditioning. The Bernoulli and empirical without-replacement curves are nearly identical and outperform the heuristic preconditioner, especially for sparse signals at intermediate values of m.}]{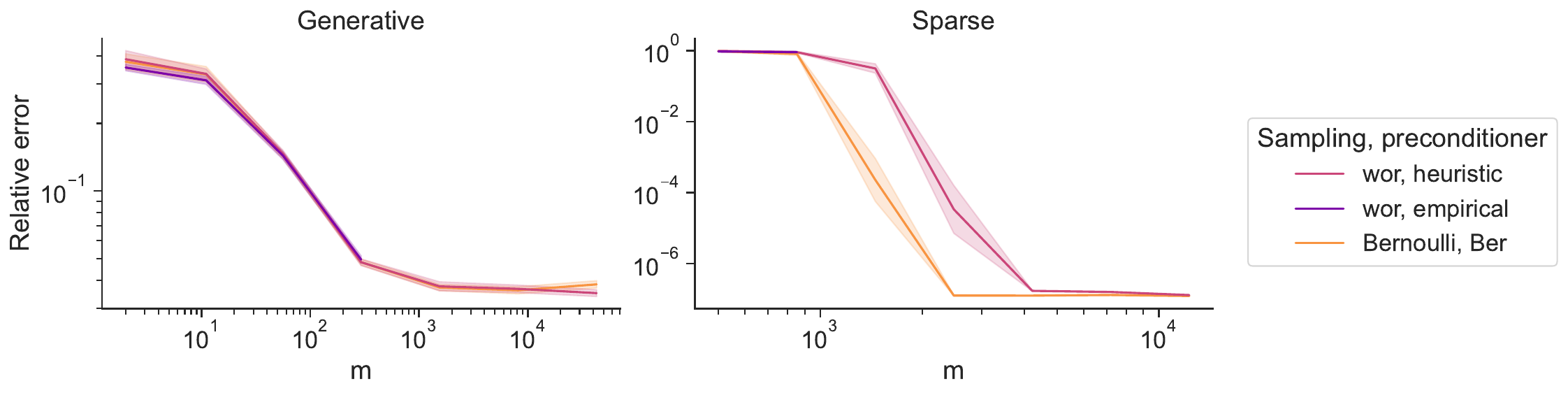}
\caption{Comparison of optimized without-replacement sampling with two different preconditioners with the optimized Bernoulli sampling scheme. The label ``wor" denotes optimized without-replacement sampling,
    with ``empirical" preconditioning as in~\cite{hoppeSamplingStrategiesCompressive2023},
    and ``heuristic" preconditioning as introduced in the text. The label ``Bernoulli, Ber" denotes the optimized Bernoulli sampling scheme with the Bernoulli preconditioner. The generative plot has 200 experiments for each data point, and
    the sparse plot has 20. We display a line for geometric mean and a band for the geometric standard error (the uncertainty of the geometric mean estimator). Note that
    in the Generative plot, the Bernoulli curve on the left plot is mostly underneath the ``wor, heuristic" curve. Sparsity level is $k = 500$ (1\%) and the code dimension of the generative model is $k = 200$ (0.5\%).
\label{fig:with_heuristic.pdf}}
\end{figure}

In \Cref{fig:with_heuristic.pdf}, we sidestep the computational difficulties of
without-replacement sampling with the empirical preconditioner of~\cite{hoppeSamplingStrategiesCompressive2023}
by using a alternative, heuristic, but easily computable
preconditioner $D$, which we now describe.
Given that we are sampling without-replacement with probabilities
$\boldsymbol{p}$ with $\sum_{i = 1}^n p_i = 1$,
we use a heuristic for the marginal sampling probabilities
of measurement vectors to be $w_i = 1 - \exp(- \lambda p_i)$
for some constant $\lambda > 0$
chosen so that $\sum_{i = 1}^n w_i = m$. We then use
the preconditioner $D$ with diagonal entries
$d_i = \sqrt{\frac{m}{ n w_i}}$. Though this preconditioner lacks theoretical guarantees,
in \Cref{fig:with_heuristic.pdf} it performs slightly better than the preconditioner from \cite{hoppeSamplingStrategiesCompressive2023}, 
in the regime where the empirical preconditioner can be efficiently computed.

One important numerical result we find is that in \Cref{fig:with_heuristic.pdf}, the optimized Bernoulli sampling scheme
outperforms the optimized without-replacement sampling
scheme in the sparse setting. We believe this occurs because the local
coherences have a heavier tail in the sparse setting
than in the generative setting. A heavy tail provides the potential of
``distracting" the optimized without-replacement sampling scheme away from
the most important measurement vectors, whereas optimized Bernoulli is immune
to this kind of distraction because of
its ability to sample the important measurement vectors
deterministically, as exemplified in the toy problem of \Cref{loc:body.toy_example}.

We note that in other experiments not presented in this paper, we found
similar performance when using the preconditioner from the optimized Bernoulli sampling scheme
paired with the optimized without-replacement sampling scheme, as when using the heuristic preconditioner introduced above. 
Additionally, in a preliminary investigation we found that it is possible
to achieve moderate performance gains by
regularizing the preconditioner with a small additive constant in the
denominator ($d_i = \sqrt{\frac{m}{n (w_i + 10^{-7}m) }}$). We do not use this constant in the experiments
presented herein for closer
correspondence to our theory.
\section{Proofs}
\label{loc:body.proofs}
\subsection{RIP of preconditioned subsampled unitary matrices}
\label{loc:body.proofs.rip_of_preconditioned_subsampled_unitary_matrices}
We begin with the case where the prior set is a single subspace.
\begin{lemma}[Deviation of Bernoulli matrix on subspace]
\label{loc:deviation_of_bernoulli_matrix_on_subspace.statement}
Let $F \in \measfield^{n \times n}$ be a unitary matrix, $S \in  \mathbb{R}^{m \times  n}$ a \hyperref[loc:bernoulli_selector_sampling_matrix.statement]{Bernoulli selector sampling matrix} with probability weights $\boldsymbol{w} \in (0,1]^n \cap m\Delta^{n-1}$. Let $\boldsymbol{\alpha} \in \mathbb{R}^n_{++}$ be the local coherences of $F$ with respect to a $\ell-$dimensional subspace $\mathcal{U} \subseteq \field^n$. Let $D \in  \mathbb{R}^n$ be a diagonal pre-conditioning matrix with diagonal entries $D_{i,i}= \sqrt{\frac{ m }{nw_i}}$. Let $\gamma$ be as defined in~\Cref{loc:tight_bernoulli_complexity.statement}. Let $t > 0$. Then
\begin{equation*}
\sup_{\boldsymbol{x} \in \mathcal{U}\cap \mathbb{S}^{n-1}} \left| \|SDF\boldsymbol{x}\|_{2}-1 \right|  \lesssim  \frac{\gamma(\boldsymbol{\alpha}, \boldsymbol{w})}{\sqrt{ m}}\sqrt{\log \ell} + \frac{\gamma(\boldsymbol{\alpha}, \boldsymbol{w})}{ \sqrt{m}} t
\end{equation*}
with probability at least $1-2\exp(-t^2)$.
\end{lemma}
\begin{proof}[\hypertarget{loc:deviation_of_bernoulli_matrix_on_subspace.proof}Proof of \Cref{loc:deviation_of_bernoulli_matrix_on_subspace.statement}]

For some $\boldsymbol{w} \in (0,1]^n \cap \Delta^{n-1}$, let $\xi_i \sim Ber(w_i),$ for $i \in [n]$, be independently distributed random variables. With a slight abuse of notation, let $S \in \{ 0, 1 \}^{n \times n}$ be a square diagonal matrix with entries $S_{i,i}= \sqrt{ \frac{n}{m} }\xi_i$. This definition differs from the Bernoulli sampling matrix of~\Cref{loc:bernoulli_selector_sampling_matrix.statement} only in the fact that we do not drop the rows that take on a value of $0$. Consider that this modification does not affect the truthfulness of \Cref{loc:deviation_of_bernoulli_matrix_on_subspace.statement} because the quantity $\lVert SDFx\rVert_2$ remains unchanged.

In what follows, we will use the fact that $\forall \boldsymbol{x} \in  \mathcal{U}$, $SDF\boldsymbol{x} = SDFP_\mathcal{U}^*P_\mathcal{U}\boldsymbol{x}$ where $P^*_\mathcal{U} \in \field^{n \times \ell}$ is a matrix whose columns make an orthonormal basis for $\mathcal{U}$. Notice that $P_\mathcal{U}^* P_\mathcal{U} = \proj_\mathcal{U} \in \field^{n \times n}$, the orthogonal projection on to $\mathcal{U}$. Now consider
\begin{align}
(\star) :=\sup_{\boldsymbol{x} \in \mathcal{U} \cap \sphere{n}} \left| \left\| SDF\boldsymbol{x}\right\|_{2}^2 - 1 \right|
 & = \sup_{\boldsymbol{x} \in \mathcal{U} \cap \sphere{n}} \left| \left\| S D F P_\mathcal{U}^* P_\mathcal{U} \boldsymbol{x}\right\|_{2}^2 - 1 \right|                         \label{eq:op_norm:1}\\
 & = \sup_{\boldsymbol{u} \in \field^\ell \cap \sphere{\ell}} \left| \left\| S D F P_\mathcal{U}^* \boldsymbol{u}\right\|_{2}^2 - 1 \right|                                   \label{eq:op_norm:2}\\
 & = \sup_{\boldsymbol{u} \in \field^\ell \cap \sphere{\ell}} \left| \boldsymbol{u}^* \left[ P_{\mathcal{U}} F^* D S^* S D F P_{\mathcal{U}}^* - I \right]\boldsymbol{u} \right|. \label{eq:op_norm:3}
\end{align}

\Cref{eq:op_norm:2}  follows from a change of variables $P_\mathcal{U} \boldsymbol{x} = \boldsymbol{u} \in \field^\ell$.  The matrix in the square brackets is Hermitian, and therefore, 
$\boldsymbol{u}^* [\ldots] \boldsymbol{u}$ is a real number (see, e.g., by \cite[Result 7.15]{axlerLinearAlgebraDone2024}). We can
therefore take the real part of the Hermitian matrix:
\begin{align}
& = \sup_{\boldsymbol{u} \in \field^\ell \cap \sphere{\ell}} \left| \boldsymbol{u}^* \mathcal{R} \left[ P_{\mathcal{U}} F^* D S^* S D F P_{\mathcal{U}}^* - I \right] \boldsymbol{u} \right| \label{eq:thing_to_bound:1}\\ 
& = \sup_{\boldsymbol{u} \in \field^\ell \cap \sphere{\ell}} \left| \boldsymbol{u}^* \mathcal{R} \left[ \sum_{i=1}^{n} \left( \frac{1}{w_i} \xi_i - 1 \right)P_\mathcal{U} F^* \boldsymbol{e}_i \boldsymbol{e}_i^* F P_\mathcal{U}^* \right]\boldsymbol{u} \right| \label{eq:thing_to_bound:2}\\ 
& =\left\| \sum_{i=1}^{n} \left( \frac{\xi_i }{ w_i} - 1 \right)\mathcal{R}[ P_\mathcal{U} F^* \boldsymbol{e}_i \boldsymbol{e}_i^* F P_\mathcal{U}^* ] \right\|. \label{eq:thing_to_bound:3}\\
\end{align}
Notice that we have a sum of independent random matrices. As the central ingredient of this proof, we make use of the Matrix Bernstein inequality to bound \Cref{eq:thing_to_bound:3}.
\begin{lemma}[Matrix Bernstein]
\label{loc:matrix_bernstein.statement}
Let $X_1, ..., X_N$ be independent, mean zero, symmetric random matrices in
$\mathbb{R}^{\ell \times \ell}$, such that $||X_i|| \leq K$ almost surely for
all $i \in [N]$. Then, for every $t\geq 0$, we have
\begin{equation*}
\mathbb{P} \left\{ \left\lVert \sum_{i=1}^N X_i \right\rVert \geq t \right\} \leq 2\ell\exp \left( - \frac{t^2/2}{\sigma^2 + Kt/3}\right),
\end{equation*}
where $\sigma^2=\lVert\sum_{i=1}^N \mathbb{E}X_i^2\rVert$.
\end{lemma}
We first find $K$. For any $i  \in [n]$,
\begin{align*}
&\left\|\left( \frac{1}{ w_i} \xi_i - 1 \right)\mathcal{R}[ P_\mathcal{U} F^* \boldsymbol{e}_i \boldsymbol{e}_i^* F P_\mathcal{U}^* ]\right\|\\
 \leq &\max\left(\frac{1-w_i}{w_i} , 1 \right) \indicator_{\{w_i < 1\}} \sup_{\boldsymbol{u} \in \field^\ell \cap \sphere{\ell}} | \boldsymbol{u}^* \mathcal{R}[ P_\mathcal{U} F^* \boldsymbol{e}_i \boldsymbol{e}_i^* F P_\mathcal{U}^* ] \boldsymbol{u}|\\
= &\sup_{\boldsymbol{u} \in \field^\ell \cap \sphere{\ell}} |  \mathcal{R}[\| \boldsymbol{e}_i^* F P_\mathcal{U}^*\boldsymbol{u} \|_2^2] | \max\left(\frac{1-w_i}{w_i} , 1 \right) \indicator_{\{w_i < 1\}} \\
 = &\alpha_i^2\max\left(\frac{1-w_i}{w_i} , 1 \right) \indicator_{\{w_i < 1\}}\\
\leq &\frac{\gamma^2(\boldsymbol{\alpha}, \boldsymbol{w})}{m}
\end{align*}
Now for $\sigma^2$, with $\boldsymbol{v}_i = P_\mathcal{U} F^* \boldsymbol{e}_i$,
\begin{align}
\sigma^2 & = \left\|\sum_{i=1}^n \mathbb{E}\left[\left( \frac{\xi_i}{w_i}  - 1 \right)^2\mathcal{R}(\boldsymbol{v}_i \boldsymbol{v}_i^*)^2\right] \right \| \\
& \le  \sum_{i=1}^n \max\left(\frac{1-w_i}{w_i} , 1 \right) \indicator_{\{w_i < 1\}}\left\| \mathcal{R}(\boldsymbol{v}_i \boldsymbol{v}_i^*)^2 \right \|.
\end{align}
To bound $\left\| \mathcal{R}(\boldsymbol{v}_i \boldsymbol{v}_i^*)^2 \right \|$,  we introduce the unit vector $\hat{\boldsymbol{y}}$ to be the normalization of $\boldsymbol{y} := \mathcal{R}[\boldsymbol{v}_i \boldsymbol{v}_i^*]\boldsymbol{x}$, and then
\begin{align*}
\boldsymbol{x}^*  \mathcal{R}[\boldsymbol{v}_i \boldsymbol{v}_i^*]\mathcal{R}[\boldsymbol{v}_i \boldsymbol{v}_i^*]\boldsymbol{x}  &= \boldsymbol{x}^*  \mathcal{R}[\boldsymbol{v}_i \boldsymbol{v}_i^*]\hat{\boldsymbol{y}}\hat{\boldsymbol{y}}^* \mathcal{R}[\boldsymbol{v}_i \boldsymbol{v}_i^*]\boldsymbol{x}\\
&=  \mathcal{R}[\boldsymbol{x}^*  \boldsymbol{v}_i \boldsymbol{v}_i^*\hat{\boldsymbol{y}}]\mathcal{R}[\hat{\boldsymbol{y}}^* \boldsymbol{v}_i \boldsymbol{v}_i^*\boldsymbol{x}]\\
&\le |\boldsymbol{x}^*  \boldsymbol{v}_i|^2 |\boldsymbol{v}_i^*\hat{\boldsymbol{y}}|^2
 \le (\boldsymbol{x}^*  \boldsymbol{v}_i \boldsymbol{v}_i^* \boldsymbol{x}) \alpha_i^2.
\end{align*}
With this,
\begin{align*}
\sigma^2 &\leq \sum_{i=1}^n (\boldsymbol{x}^*  \boldsymbol{v}_i \boldsymbol{v}_i^* \boldsymbol{x}) \max_{j \in [n]} \alpha_j^2 \max\left(\frac{1-w_j}{w_j} , 1 \right) \indicator_{\{w_j < 1\}}
= \frac{\gamma^2(\boldsymbol{\alpha},\boldsymbol{w})}{m},
\end{align*}
because $\sum_{i = 1}^n \boldsymbol{v}_i \boldsymbol{v}_i^*  =  I$. Then applying Matrix Bernstein yields
\begin{equation*}
\mathbb{P}\left\{ \sup_{\boldsymbol{x} \in \mathcal{U} \cap \sphere{n}} \left| \left\| SDF\boldsymbol{x}\right\|_{2}^2 - 1 \right| \ge t \right\}
\leq 2\ell \exp\left( - \frac{m}{\gamma^2}\min\left(t^2, t\right) \right).
\end{equation*}
The result then follows from algebraic manipulations of this tail inequality (for the details, see, e.g., the proof of \cite[Lemma A.1]{planDenoisingGuaranteesOptimized2025}).
\end{proof}
\begin{proof}[\hypertarget{loc:rip_of_bernoulli_selector_non:uniform_sampling_matrix_on_a_union_of_subspaces.proof}Proof of \Cref{loc:rip_of_bernoulli_selector_non:uniform_sampling_matrix_on_a_union_of_subspaces.statement}]

Perform a union bound on
\Cref{loc:deviation_of_bernoulli_matrix_on_subspace.statement}
applied to each subspace $\mathcal{U}$
making up $\mathcal{T}$. Additional details on this union bound can be found in the proof of \cite[Lemma A.1]{planDenoisingGuaranteesOptimized2025}, which is similar.
\end{proof}
\subsection{Bounding the noise complexity}
\label{loc:body.proofs.bounding_the_noise_complexity}
\begin{proof}[\hypertarget{loc:tail_on_the_noise_sensitivity_for_adapted_bernoulli_sampling.proof}Proof of \Cref{loc:tail_on_the_noise_sensitivity_for_adapted_bernoulli_sampling.statement}]

We combine two upper-bounds.

\textbf{First upper-bound}

It is that of \Cref{eq:first:trunc} in~\Cref{loc:simple_bound_on_the_gaussian_noise_error_factor.statement};
\begin{equation*}
\|\tilde{D} \trunc(SD \boldsymbol{\alpha})\|_2^2 \leq \max(\boldsymbol{d})^2 = L\sqrt{ \frac{1}{n \min(\boldsymbol{\alpha})} }.
\end{equation*}

\textbf{Second upper-bound}

With $U = \{i  \in [n]: w^{\circ}_i < 1\}$,
\begin{equation*}
\|\tilde{D} \trunc(SD \boldsymbol{\alpha})\|_2^2 =  
\|\tilde{D}|_U \trunc(SD \boldsymbol{\alpha})|_U\|_2^2+\|\tilde{D}|_{U^c} \trunc(SD \boldsymbol{\alpha})|_{U^c}\|_2^2.
\end{equation*}
We treat the two terms differently.

\textbf{Term with unsaturated entries}

Let $\bar{\boldsymbol{w}}  \in \mathbb{R}^n_{++}$ be with entries $\bar{w}_i  = \frac{m \alpha_i^2}{L^2}$, which is such that $\bar{\boldsymbol{w}}  \ge \boldsymbol{w}^{\circ}$ entry-wise. Let $\boldsymbol{h}$ be similar to $\boldsymbol{d}$, only that it depends on $\bar{\boldsymbol{w}}$ instead of $\boldsymbol{w}^{\circ}$, i.e.,
$h_i  :=  \sqrt{\frac{m}{n \bar{w}_i}}$, $H := \Diag(\boldsymbol{h})$, and $\tilde{H} := \Diag(\sqrt{\frac{m}{n}}S\boldsymbol{h})$.
Then $\|\tilde{H} \trunc(SH \boldsymbol{\alpha})\|_2^2 \leq \|\tilde{D}|_U \trunc(SD \boldsymbol{\alpha})|_U\|_2^2$. So,
\begin{equation*}
(\star):=\|\tilde{H} \trunc(SH \boldsymbol{\alpha})\|_2^2  \le \|S H^2 \boldsymbol{\alpha}\|_2^2.
\end{equation*}
Because $h_i = \sqrt{ \frac{m}{nw^{\circ}_i} } = \frac{L}{\sqrt{ n } \alpha_i}$, and so $H^2 \boldsymbol{\alpha} =  \frac{L}{\sqrt{ n }}  \boldsymbol{h}$, we get the upper-bound
\begin{equation*}
(\star)  \le \frac{L^2}{n}\|S\boldsymbol{h}\|_2^2.
\end{equation*}
Then by Markov's inequality,
\begin{equation*}
\mathbb{E} \|S \boldsymbol{h}\|_2^2 = \frac{n}{m} \sum_{i = 1}^n w^{\circ}_i h_i^2 = \frac{m}{n}\sum_{i = 1}^n w^{\circ}_i \frac{m}{n w^{\circ}_i} = n.
\end{equation*}
This gives us that with probability at least $1-\delta$,
\begin{equation*}
(\star)  \le \frac{L^2}{\delta}.
\end{equation*}

\textbf{Saturated term}

Now for the saturated terms, with $\boldsymbol{d}$, $D$, $\tilde{D}$ back
to being defined with $\boldsymbol{w}^{\circ}$, notice that $\tilde{D} = \sqrt{\frac{m}{n}}I$,
and that $\trunc(SD \boldsymbol{\alpha})^{.2}$ are convex coefficients, so that
\begin{equation*}
\|\tilde{D}|_{U^c} \trunc(SD \boldsymbol{\alpha})|_{U^c}\|_2^2  \le \frac{m}{n}
\end{equation*}
Combining these upper-bounds, result follows.
\end{proof}
\subsection{Union bounds}
\label{loc:body.proofs.union_bounds}
\begin{proof}[\hypertarget{loc:cs_with_bernoulli_and_denoising_on_unions_of_subspaces.proof}Proof of \Cref{loc:cs_with_bernoulli_and_denoising_on_unions_of_subspaces.statement}]

Let $t_1>0$, and
\begin{equation*}
m \gtrsim \gamma^2(\boldsymbol{\alpha}, \boldsymbol{w})(\log (\ell) + \log M + t_1^2).
\end{equation*}
Then each of the following statements holds individually with probability at least $1- 2\exp(-t^2)$ for a variable $t>0$ defined within each of the results.

\begin{enumerate}
\item The matrix $SDF$ has the RIP on $\mathcal{T}$ thanks to \Cref{loc:rip_of_bernoulli_selector_non:uniform_sampling_matrix_on_a_union_of_subspaces.statement}.
\item When \emph{1.} is satisfied, the recovery error is bounded as specified by \Cref{loc:signal_recovery_with_subsampled_unitary_matrix_with_gaussian_noise_on_union_of_subspaces.bernoulli}.
\end{enumerate}

We distinguish between the variables $t$ used within each of the two statements by
re-labelling them $t_{1}, t_{2}$ respectively. For some $\delta>0$, let
$2\exp(-t_1^2) = \frac{1}{2}\delta$ and $2\exp(-t_{2}^2) = \frac{1}{2}\delta$.
Then $t_{1} = t_2 = \sqrt{ \log \frac{4}{\delta} }$. The fact that the second
statement is conditional on the success of the first only lessens the true
probability of failure, and so the probability of failure is no more than
$\frac{1}{2}\delta + \frac{1}{2}\delta = \delta$. The result follows.
\end{proof}
\begin{proof}[\hypertarget{loc:optimized_bernoulli_cs_on_union_of_subspaces.proof}Proof of \Cref{loc:optimized_bernoulli_cs_on_union_of_subspaces.statement}]

The proof is a minor variation on \hyperlink{loc:cs_with_bernoulli_and_denoising_on_unions_of_subspaces.proof}{the proof} of \Cref{loc:cs_with_bernoulli_and_denoising_on_unions_of_subspaces.statement}.

Let $t_1>0$, and let
\begin{equation*}
m \gtrsim L(\boldsymbol{\alpha}, m)^2 \left( \log \ell+ \log M + t_1^2\right).
\end{equation*}
Each of the following statements holds individually with probability at least $1- 2\exp(-t^2)$ for a variable $t>0$ defined within each of the three results.
\begin{enumerate}
\item The matrix $SDF$ has the RIP thanks to \Cref{loc:rip_of_bernoulli_selector_non:uniform_sampling_matrix_on_a_union_of_subspaces.statement}, by bounding $\gamma$ by $\eta$ and evaluating at the optimized probability weights $\boldsymbol{w}^\circ$.
\item When \emph{1.} is satisfied, the recovery error is bounded as specified by \Cref{loc:signal_recovery_with_subsampled_unitary_matrix_with_gaussian_noise_on_union_of_subspaces.bernoulli}.
\item When \emph{1.} is satisfied, the noise sensitivity $\|\widetilde{D}\trunc(SD \boldsymbol{\alpha})\|_2$ in the recovery error bound of~\Cref{loc:signal_recovery_with_subsampled_unitary_matrix_with_gaussian_noise_on_union_of_subspaces.bernoulli} is bounded thanks to~\Cref{loc:tail_on_the_noise_sensitivity_for_adapted_bernoulli_sampling.statement}.
\end{enumerate}

We distinguish between the variables $t$ used within each of the three statements by re-labelling them $t_{1}, t_{2}, t_{3}$ respectively. For some $\delta>0$, let $2\exp(-t_1^2) = \frac{1}{10}\delta$, $2\exp(-t_{2}^2) = \frac{1}{10}\delta$, and $t_{3} = \frac{8}{10}\delta$. Then $t_{1} = t_2 = \sqrt{ \log \frac{20}{\delta} }$. The required statement then holds with probability at least $1-\delta$ from a union bound on the three statements above for this choice of $t_1, t_2, t_3$.
\end{proof}
\subsection{Optimizing the sampling scheme}
\label{loc:body.proofs.optimizing_the_sampling_scheme}
We begin with a lemma for a larger class of optimization problems.
\begin{lemma}[On the element wise maximized objectives with soft boundaries]
\label{loc:on_the_element_wise_maximized_objectives_with_soft_boundaries.statement}
Take a number of decreasing functions $f_1, \dots, f_n$, which are of the form $f_i:(0,1] \to \mathbb{R} \forall  i \in [n]$. For $\beta > 0$, consider the optimization problem
\begin{equation*}
\minimize \max_{j \in  [n]} f_j(w_j) \text{ s.t. } \boldsymbol{w} \in  \beta\Delta^{n-1}.
\end{equation*}
We have that:
\begin{enumerate}
\item Any point $\boldsymbol{w}^* \in \beta\Delta^{n-1}$ such that
\begin{equation}
\label{eq:optfirst}
\forall i \in [n], \quad  \lim_{w_i \uparrow w^*_i} f_i(w_i) \geq \max_{j \in  [n]} f_j(w^*_j)
\end{equation}
is a minimizer. Here, ``$\uparrow$" denotes a left limit.
\item If the functions $f_1, \dots, f_n$ are moreover \emph{strictly} decreasing, then the minimizer $\boldsymbol{w}^*$ in (1.) is unique.
\end{enumerate}
\end{lemma}
\begin{proof}[\hypertarget{loc:on_the_element_wise_maximized_objectives_with_soft_boundaries.proof}Proof of \Cref{loc:on_the_element_wise_maximized_objectives_with_soft_boundaries.statement}]

Let $\boldsymbol{w}^* \in \beta \Delta^{n-1}$ satisfy \Cref{eq:optfirst}. We argue that it is impossible to find another point $\hat{\boldsymbol{w}} \neq \boldsymbol{w}^*$ in $\beta \Delta^{n-1}$ such that 
\begin{equation}
\label{eq:contramin}
\max_{j \in  [n]} f_j(\hat{w}_j) <  \max_{j \in  [n]} f_j(w^*_j).
\end{equation}
Since $\hat{\boldsymbol{w}} \neq \boldsymbol{w}^*$, and since they both lie in $\beta \Delta^{n-1}$, there is an $i \in [n]$ such that $\hat{w}_i < w^*_i$. Then
\begin{equation}
\label{eq:lim}
f_i(\hat{w}_i) \geq \lim_{w_i \uparrow w^*_i} f_i(w_i) \geq \max_{j \in [n]} f_j(w^*_j),
\end{equation}
where the second inequality follows from \Cref{eq:optfirst}. But this yields a contradiction with \Cref{eq:contramin}, so $\boldsymbol{w}^*$ must be a minimizer.

The second part of the result follows from noting that when the functions are strictly decreasing, the first inequality in \Cref{eq:lim} is strict.
\end{proof}
With \Cref{loc:on_the_element_wise_maximized_objectives_with_soft_boundaries.statement} we prove the results of~\Cref{loc:body.proofs.optimizing_the_sampling_scheme}.
\begin{proof}[\hypertarget{loc:optimize_simple_bernoulli_selector_sampling.proof}Proof of \Cref{loc:optimize_simple_bernoulli_selector_sampling.statement}]

We use 
\Cref{loc:on_the_element_wise_maximized_objectives_with_soft_boundaries.statement}
with the functions $f_i(x) := \alpha_i \sqrt{ \frac{m}{x}}\indicator_{\{x < 1\}}$, which are strictly decreasing for all $i  \in [n]$, and which satisfy $\max_{i \in [n]} f_i(w^{\circ}_i) = \eta(\boldsymbol{\alpha}, \boldsymbol{w}^{\circ})$. Since $\eta(\boldsymbol{\alpha}, \boldsymbol{w}^{\circ}) = L$, it remains only to show that $\forall  i \in [n], \lim_{w_i \uparrow w^{\circ}_i} f_i(w_i) \geq L$.

For any fixed \emph{unsaturated} entry $i \in [n]$, $f_i$ is continuous in a neighborhood of $w^{\circ}_i$, so it is sufficient to show that $f_i(w^{\circ}_i) \ge L$. This holds because with the formula $w^{\circ}_i =\frac{m \alpha^2_j}{L^2}$, one can compute that $f_i(w^{\circ}_i)=L$. 

Second, for any \emph{saturated} entry $i \in [n]$, we know from the definition of $\boldsymbol{w}^{\circ}$ that $\frac{\alpha_i^2 m}{L^2(\boldsymbol{\alpha}, m)} \ge 1$, and therefore,
\begin{equation*}
\lim_{w_i \uparrow 1} f_i(w_i) = \lim_{w_i \uparrow 1} \alpha_i\sqrt{\frac{m}{w_i}} = \alpha_i \sqrt{m} \ge L.
\end{equation*}
\end{proof}
\subsection{Properties of the optimized sampling scheme}
\label{loc:body.proofs.properties_of_the_optimized_sampling_scheme}
\begin{proof}[\hypertarget{loc:norm_of_the_optimized_probability_weights.proof}Proof of \Cref{loc:norm_of_the_optimized_probability_weights.statement}]

WLOG let $\boldsymbol{\alpha}$ be increasing (otherwise, re-index it).
Note that $J$ is well-defined because the set over which we take a maximum is non-empty. Indeed,
\begin{equation*}
\frac{m \alpha_1^2}{R^2(1; \boldsymbol{\alpha}, m)} = 1 - (n - m) <  1,
\end{equation*}
because $m < n$. 

We show that the index $J$ identifies the last unsaturated entry of the optimized weights. The entry $J$ is unsaturated because
\begin{equation*}
w_J^\circ = \frac{m \alpha_J^2}{L^2(\boldsymbol{\alpha}, m)}  = \frac{m
\alpha_J^2}{R^2(J;\boldsymbol{\alpha},  m)}<1
\end{equation*}
by definition of $J$. If $J = n$, $J$ is indeed the last unsaturated entry. If $J < n$, then the entry $J+1$ satisfies
\begin{align*}
&\frac{m \alpha_{J+1}^2}{R^2(J+1; \boldsymbol{\alpha}, m)}  \ge 1 \\
 \implies &\frac{\alpha_{J+1}^2}{\left(\frac{\|\boldsymbol{\alpha}|_{\le J+1}\|_2^2}{m-(n-(J+1))}\right)} \ge 1\\
\implies  &\alpha_{J+1}^2(J-(n-m)+1) \ge \|\boldsymbol{\alpha}|_{\le J+1}\|_2^2\\
\implies &\alpha_{J+1}^2(J-(n-m)) \ge \|\boldsymbol{\alpha}|_{\le J}\|_2^2\\
\implies &\frac{\alpha_{J+1}^2}{\left(\frac{\|\boldsymbol{\alpha}|_{\le J}\|_2^2}{(J-(n-m))}\right)} \ge 1\\
\implies &\frac{m \alpha_{J+1}^2}{L^2(\boldsymbol{\alpha}, m)} \ge 1.
\end{align*}
and so $w_{J+1}^\circ = 1$. Since the expression
$\frac{m\alpha_j^2}{L^2(\boldsymbol{\alpha}, m)}$ is monotone in $j$, this means
that $w_j^\circ = 1$ for $j>J$ and $w_j^\circ<1$ for $j \le J$.

Then
\begin{align*}
\sum_{j = 1}^n w_j^\circ & = \sum_{j = 1}^{J}
\frac{m\alpha_j^2}{L^2(\boldsymbol{\alpha}, m)} + \sum_{j = J+1}^n 1\\
& = \frac{m(J-(n-m))}{m\|\boldsymbol{\alpha}|_{[J]}\|_2^2} \left(\sum_{i = 1}^J \alpha_j^2\right) + n-J\\
& = m.
\end{align*}
\end{proof}
\begin{proof}[\hypertarget{loc:upper_bound_bernoulli_l_with_local_coherences.proof}Proof of \Cref{loc:upper_bound_bernoulli_l_with_local_coherences.statement}]

We begin by showing that the upper-bound holds. Consider $f(x) := \sum_{i=1}^n \min\left(\frac{\alpha^2_i}{x}, 1\right) - m$ and 
$g(x):= \sum_{i = 1}^n \frac{\alpha_i^2}{x} - m$. They are strictly decreasing functions
and $\forall x \in  \mathbb{R}, f(x) \le  g(x)$. It follows that
the root of $g$ must be greater than the root of $f$. By inspection,
the root of $f$ is $L$ and the root of $g$ is $\|\boldsymbol{\alpha}\|_2$,
and so the upper-bound in the statement follows.

For the lower bound, first we recall the definition of the variable $J$ in~\Cref{loc:optimize_simple_bernoulli_selector_sampling.statement} to be $J := \max\left\{  J \in [n]: \frac{\lVert\boldsymbol{\alpha}|_{< J}\rVert_2^2}{\alpha_{J}^2} >   (m - (n-J) - 1) \right\}.$ Then we see that $J \geq n-m+1$ since if we let $J=n-m+1$, the r.h.s. in the inequality in the definition of $J$ is zero, and the l.h.s. is always strictly positive since we assumed strictly positive local coherences.
Then the lower-bound holds because
\begin{equation*}
L^2 := \lVert \boldsymbol{\alpha}|_{\leq J}\rVert_2^2\frac{m}{(m - (n-J))} \ge \lVert \boldsymbol{\alpha}|_{\leq J}\rVert_2^2 \ge \lVert \boldsymbol{\alpha}|_{\leq n-m+1}\rVert_2^2.
\end{equation*}
\end{proof}
\begin{proof}[\hypertarget{loc:monotonicity_of_l_in_m.proof}Proof of \Cref{loc:monotonicity_of_l_in_m.statement}]

For some fixed $m  \in \mathbb{N}$ and local coherences $\boldsymbol{\alpha}  \in \mathbb{R}^n_{++}$, the optimized probability $\boldsymbol{w}^*$ from \Cref{loc:optimize_simple_bernoulli_selector_sampling.statement} is contained in $m\Delta^{n-1}$. Therefore, 
\begin{equation*}
\sum_{i = 1}^n w_i^*  =  \sum_{i = 1}^n \min\left( \frac{m \alpha^2_j}{L^2(m)}, 1 \right) = m,
\end{equation*}
where denote $L$ as a function of $m$.
We solve for its first derivative $L'(m)$ by implicit differentiation. Differentiating by $m$ and isolating $L'(m)$, we find that
\begin{equation*}
\text{sgn}(L') = \text{sgn}\left(1 - \frac{L^2}{\|\boldsymbol{\alpha}|_{ \le J}\|_2^2}\right).
\end{equation*}
That $L  \ge  \|\boldsymbol{\alpha}|_{ \le J}\|_2$ follows from \hyperlink{loc:upper_bound_bernoulli_l_with_local_coherences.proof}{the proof} of \Cref{loc:upper_bound_bernoulli_l_with_local_coherences.statement}.
\end{proof}

\section*{Acknowledgements}
Large language models were used while writing the manuscript for help with grammar and phrasing (Claude, Grok, and Chatgpt).
\section*{Funding}
\label{loc:body.aknowledgements}
Y. Plan is partially supported by an NSERC Discovery Grant (GR009284), an NSERC Discovery Accelerator Supplement (GR007657), and a Tier II Canada Research Chair in Data Science (GR009243). O. Yilmaz was supported by an NSERC Discovery Grant (22R82411) O. Yilmaz also acknowledges support by the Pacific Institute for the Mathematical Sciences (PIMS) and the CNRS -- PIMS International Research Laboratory. 
\section*{Data availability}
The data underlying the numerical experiments in this article are available in the CELEBA dataset~\cite{liuDeepLearningFace2015} and the flower dataset~\cite{nilsbackAutomatedFlowerClassification2008}. No new datasets were generated for this study.

\bibliographystyle{abbrvnat}
\bibliography{bibliography}

\appendix
\renewcommand{\theHsection}{appendix.\thesection}

\section{Equivalence of duplicate-rejection sampling and sequential renormalized sampling}
\label{app:equivalence_sampling}

\begin{proposition}[Equivalence of two without-replacement implementations]
\label{prop:equiv_sampling}
Let $F=\{f_1,\dots,f_n\}$ and let $p=(p_1,\dots,p_n)$ satisfy $p_j\ge 0$ and $\sum_{j=1}^n p_j=1$. Fix $1\le m\le n$. Consider the following two procedures that generate an ordered $m$-tuple of distinct indices $(I_1,\dots,I_m)$.

\begin{enumerate}
\item[(1)] \textbf{Duplicate-rejection (with-replacement) sampling.}
Draw i.i.d.\ random variables $X_1,X_2,\dots$ with $\mathbb P(X_t=j)=p_j$.
Accept a draw if it has not appeared previously among accepted values; otherwise reject it and continue.
Stop after $m$ distinct values have been accepted, and denote the accepted sequence by $(I_1,\dots,I_m)$.

\item[(2)] \textbf{Sequential sampling without replacement with renormalization.}
Set $S_0=\varnothing$. For $r=1,\dots,m$, sample $I_r$ from $\{1,\dots,n\}\setminus S_{r-1}$ according to
\[
\mathbb P(I_r=j \mid S_{r-1})
=
\frac{p_j}{\sum_{\ell\notin S_{r-1}} p_\ell},
\qquad j\notin S_{r-1},
\]
and set $S_r=S_{r-1}\cup\{I_r\}$.
\end{enumerate}

Then the two procedures induce the same distribution on $(I_1,\dots,I_m)$.
\end{proposition}

\begin{proof}
It suffices to show that, in procedure \textup{(1)}, conditional on the current accepted set being $S\subset\{1,\dots,n\}$, the next accepted index has the same conditional distribution as in \textup{(2)}.

In \textup{(1)}, let $q=\sum_{i\in S}p_i$. The next accepted value is the first draw outside $S$. For any $j\notin S$,
\[
\mathbb P(\text{next accepted}=j)
=
\sum_{t=1}^\infty q^{\,t-1}p_j
=
\frac{p_j}{1-q}
=
\frac{p_j}{\sum_{\ell\notin S}p_\ell}.
\]
Thus, conditional on $S$, the next accepted index is drawn from $S^c$ with probabilities proportional to $p_j$, exactly as specified in \textup{(2)}.

Since the first accepted index has distribution $p$ in both procedures and the same conditional rule governs each subsequent step, the joint law of $(I_1,\dots,I_m)$ coincides by induction.
\end{proof}

\section{Python code}
\label{loc:body.proofs.python_code}

\begin{lstlisting}[language=Python, caption={Computation of the optimized Bernoulli probability weights from a local coherence vector.}, label={lst:optimized-bernoulli}]
def Optimized_Bernoulli_prob_weights(loc_co: NDArray, m: int):
    pos_co = loc_co[loc_co > 0]
    n = pos_co.size
    assert m <= n, "Cannot sample enough measurements."
    if m == n:
        samplings = np.zeros_like(loc_co)
        samplings[loc_co > 0] = 1
        return samplings, 0, 0.0
    sorted_co = np.sort(pos_co.flat)

    def Lsqrd(j, m, n, sorted_co):
        assert j > n - m - 1
        assert m <= n
        assert j <= n - 1
        return m * np.sum(sorted_co[:j + 1] ** 2) / (j - (n - m - 1))

    J = None
    for j in range(n - 1, n - m - 1, -1):
        if m * sorted_co[j] ** 2 < Lsqrd(j, m, n, sorted_co):
            J = j
            break
    if J is None:
        raise ValueError("No valid J found (should never happen).")
    Lsqr = Lsqrd(J, m, n, sorted_co)
    return np.clip(m * loc_co ** 2 / Lsqr, 0, 1), J, Lsqr
\end{lstlisting}
\end{document}